\documentclass[10pt]{amsart}
\UseRawInputEncoding
\usepackage{amsfonts}
\usepackage{amsmath}
\usepackage{amssymb}
\usepackage{graphicx}%
\usepackage{dcolumn}
\usepackage{bm}
\usepackage{pdfsync}
\usepackage{mathrsfs}
\usepackage{enumitem}
\usepackage{stmaryrd}
\usepackage[colorlinks,linkcolor=blue,citecolor=blue,urlcolor=blue]{hyperref}
\numberwithin{equation}{section}

\newcommand{\beq}{\begin{equation}}
\newcommand{\eeq}{\end{equation}}
\newcommand{\beqa}{\begin{eqnarray}}
\newcommand{\eeqa}{\end{eqnarray}}
\newcommand{\bc}{\begin{cases}}
\newcommand{\ec}{\end{cases}}
\newcommand{\nn}{\nonumber}

\newtheorem{conjecture}{Conjecture}
\newtheorem{definition}{Definition}
\newtheorem{proposition}{Proposition}
\newtheorem{theorem}{Theorem}
\newtheorem{corollary}{Corollary}
\newtheorem{lemma}{Lemma}
\theoremstyle{definition}

\newtheorem{remark}{\textbf{Remark}}

\begin{document}

\subjclass[2020]{MSC: 39A05, 05A40, 34G20}
\title[Categorical Discretization of Dynamical Systems via Rota  Algebras]{Laurent sequences, extended Rota Algebras \\ and Categorical Discretization \\ of Dynamical Systems }

\author{Miguel A. Rodr\'iguez}
\address{Departamento de F\'{\i}sica Te\'{o}rica, Facultad de F\'{\i}sicas, Universidad Complutense de Madrid, 28040 -- Madrid, Spain}
\email{rodrigue@ucm.es}

\author{Piergiulio Tempesta}
\address{Instituto de Ciencias Matem\'aticas (CSIC-UAM-UCM-UC3M), C/ Nicol\'as Cabrera, No 13--15, 28049 Madrid, Spain}
\address{Departamento de F\'{\i}sica Te\'{o}rica, Facultad de F\'{\i}sicas, Universidad Complutense de Madrid, 28040 -- Madrid, Spain}
\email{p.tempesta@fis.ucm.es, piergiulio.tempesta@icmat.es}

\keywords{Differential equations, integrable maps, Galois differential algebras, category theory, recurrences, lattice field theories 
}

\begin{abstract}
We introduce a novel integrability-preserving discretization for a broad
class of differential equations with variable coefficients, encompassing
both linear and nonlinear cases. The construction is achieved via a categorical
approach that enables a unified treatment of continuous and discrete dynamical
systems.

Our theoretical framework is grounded on a generalization of G. C.
Rota's finite operator calculus, which enables us to extend the theory
of basic sequence of polynomials to the setting of Laurent polynomials.
Accordingly, we introduce the notion of an
\textit{extended Rota algebra}, defined as a Galois differential algebra
in which all difference operators act as derivations on the space of Laurent
power series with respect to a suitably defined functional product.

The core of our theory relies on the existence of covariant functors between
the newly proposed Rota category of Galois differential algebras and suitable
categories of abstract dynamical systems.

In this setting, under certain regularity assumptions, a differential equation
and its discrete analogues are naturally interpreted as objects of the
same category. This perspective enables the construction of a vast class
of integrable maps that share with their continuous analogues a wide set
of exact solutions, \textit{regular} or \textit{singular} and, in the linear
case, the Picard-Vessiot group.
\end{abstract}

\date{July 2, 2026}
\maketitle

\tableofcontents

\vspace{3mm}

\section{Introduction}
\label{sec1}

\subsection{Historical background}
\label{sec1.1}

The discretization of dynamical systems in an
\textit{integrability preserving} way has been the subject of intensive
investigation over the past decades \cite{Suris},
\cite{VeselovRMS1991}, owing to its considerable impact across a wide range
of fields, including numerical analysis, discrete mathematics, algorithm
theory, quantum physics and statistical mechanics. The guiding principle
in this line of research is that a meaningful discretization should not
merely approximate the continuous dynamics, but rather reproduce its underlying
geometric and algebraic structures. In particular, the preservation of
first integrals, Poisson or symplectic structures, and Lax representations
has emerged as a key requirement for ensuring that the discrete evolution
retains the qualitative features of the original system. This viewpoint
has led to the development of geometric and structure-preserving integrators,
as well as to intrinsically discrete models that are integrable in their
own right.

In this framework, discrete differential geometry and its associated algebraic
structures have played an increasingly prominent role \cite{AL1}, \cite{BS1},
\cite{BS2}, \cite{GK}, \cite{Nov1}, \cite{Nov2}, \cite{Nov3},
\cite{LNP}, \cite{Kuper1}, \cite{MV}. A central insight is that integrability
can often be reformulated in purely discrete terms, for instance through
multidimensional consistency, zero-curvature representations on lattices,
or compatibility conditions for discrete connections. These approaches
reveal that many integrable hierarchies admit natural discretizations that
are not ad hoc approximations, but instead arise from deeper combinatorial
or algebraic principles. In particular, the interplay between discrete
geometry, Poisson-Lie structures, and algebraic frameworks such as cluster-type
transformations has provided a unifying language to describe a broad class
of integrable discrete systems.

At the same time, the need for structurally consistent discretizations
is strongly motivated by applications in fundamental physics. Lattice formulations
of field theories \cite{MM}, together with modern approaches to quantum
mechanics \cite{FL} and quantum gravity \cite{ash}, \cite{thooft},
\cite{RS}, require discrete counterparts of continuum equations in which
key properties, such as gauge invariance, conservation laws, or integrability
features, are retained at the discrete level. In such contexts, discretization
is not simply a computational device, but a conceptual tool that may encode
fundamental aspects of the underlying theory, particularly when a continuum
description is expected to emerge only in a suitable limit.

Closely related to these developments is the problem of symmetry-preserving
discretization of continuous models, which has also attracted considerable
attention (see, e.g., \cite{LWYbook} and references therein). Here, the
aim is to construct discrete evolutions that exactly inherit the symmetry
group of the original equations, thereby ensuring the preservation of invariant
manifolds and structural constraints. This perspective is deeply intertwined
with integrability, since many integrable systems are characterized by
rich symmetry algebras, and their preservation at the discrete level is
often essential for maintaining solvability.

Within this general setting, the approach developed in this paper is motivated
by the idea that integrability-preserving discretizations can be systematically
constructed by merging category theory with the finite operator theory
developed by Roman, Rota and collaborators. Our results point towards a
framework in which discrete dynamics emerge naturally from structural considerations,
allowing one to retain integrability properties in a transparent and conceptually
unified way.

The central feature of our discretization scheme is that all operators
involved--namely, the delta operators of the Roman-Rota formalism--act
as \textit{derivations} on suitable Galois function algebras. Consequently,
many structural properties of categorically equivalent equations are preserved.
In particular, a large class of solutions is retained. We emphasize that, to the best of our knowledge, the classes of difference
equations arising from our approach are novel, and not related to those
obtained with other discretization schemes.

\subsection{Statement of the problem}
\label{sec1.2}

The purpose of this paper is to introduce a unified theoretical framework
which allows us to address and solve, in a quite general sense and from
a novel perspective, the longstanding problem of the integrability-preserving
discretization of a large class of dynamical systems, both
\textit{linear and nonlinear}, even in the case of variable coefficients.
These coefficients are expressed in terms of either real
\textit{analytic} functions (regular case), or Laurent polynomials with
\textit{polar singularities} (singular case).

Precisely, the main theorems of this work ensure the integrability-preserving
discretization of two classes of differential equations with variable coefficients:

\noindent
a) Linear ODEs of the form
\begin{equation}
a_{N}(x) \frac{d^{N}}{dx^{N}}y + a_{N-1}(x) \frac{d^{N-1}}{dx^{N-1}}y +
\cdots +a_{1}(x) \frac{d}{dx} y+a_{0}(x) y +g(x)=0 \ ;
\label{lincont}
\end{equation}
b) Nonlinear ODEs of the form
\begin{equation}
\frac{d^{m}}{dx^{m}}y= c_{N}(x) y^{N}+c_{N-1}(x) y^{N-1}+\cdots +c_{1}(x)
y+ c_{0}(x) \ .
\label{nonlincont}
\end{equation}
Here $y: \mathbb{R}\to \mathbb{R}$ is the dependent, scalar variable;
$N, m\in \mathbb{N}\setminus \{0\}$,
$a_{0}(x),\ldots , a_{N}(x), g(x)$ and $c_{0}(x),\ldots , c_{N}(x)$ are
real analytic functions, or Laurent polynomials.

For each of the families of equations \eqref{lincont} and \eqref{nonlincont}, we will provide two distinct discrete analogs, designed
to preserve either regular or singular solutions. All resulting equations
are defined on a regular mesh of points. In the continuum limit, these
difference equations converge to their corresponding differential equations,
thereby providing a \textit{discrete approximation} of the original continuous
models.

The proposed approach combines the theory of Galois differential algebras
and an extension of the theory of finite difference operators in the formulation
given by G.-C. Rota \cite{Rota}, and S. Roman \cite{Roman}, with methods
from category theory. Within our framework, a continuous dynamical system
and its infinitely many discrete counterparts are regarded as distinct
representations of an underlying abstract equation, and are naturally interpreted
as objects in a category of equations associated with the given system.
The morphisms of this category correspond to the various discrete realizations
of the abstract equation. Accordingly, the resulting difference equations
are said to be \textit{categorically equivalent}.

The theoretical framework developed in the present work allows us to substantially
generalize the results of \cite{TempestaJDE}, where discretizations were
restricted to nonlinear dynamical systems, and with \textit{constant} coefficients
only. In particular, our approach provides the first general treatment
of discretizations for nonlinear systems with variable coefficients. Moreover,
for the novel case of linear dynamical systems, we extend the Frobenius-type
theorem presented in \cite{RRT2025} -- originally formulated for second-order
equations -- to encompass differential equations of arbitrary order with
variable coefficients. Another significant novelty concerns the preservation
of \textit{singular solutions} of polar type, not just regular ones as in
the previous studies. This property is achieved through the introduction
of suitable \textit{Laurent basic sequences of polynomials}, a concept developed
in the present work.

\subsection{Extended Rota algebras}
\label{sec1.3}

We shall introduce the notion of
\textit{extended Rota differential algebra} as a pair
$\left (\mathcal{F_{\mathcal{L}}},\mathcal{Q}\right )$ where
$(\mathcal{F_{L}}, +,\cdot , *)$ is an algebra of formal Laurent series
on a mesh $\mathcal{L}$ of points, endowed with an associative and commutative
product ``*''; besides, $\mathcal{Q}$ is a \textit{delta operator}
\cite{Rota} that acts as a derivation with respect to this product. This
notion generalizes that of Rota algebra, introduced in
\cite{TempestaJDE} and further studied in \cite{RRT2025},
\cite{RT2024}, where standard formal power series were considered.

In particular, this standpoint allows one to consider ordinary difference operators
as derivations under suitable function products. The idea of an adapted
product of polynomials, ensuring for difference operators the validity
of the Leibniz rule, was proposed in the important papers \cite{Ward},
\cite{BF}.

An analogous polynomial product has appeared, in a different context, in
the theory of linear operators acting on polynomial spaces
\cite{ismail}. The collection of all Rota algebras
$\mathcal{R} (\mathcal{F}) $ represents a subcategory of the category of
associative algebras \cite{Mac Lane}.

\subsection{A categorical discretization}
\label{sec1.4}

Given a continuous dynamical system of the form  \eqref{lincont} or \eqref{nonlincont}, we first select a specific difference operator to serve
as the discrete derivative, imposing the additional requirement that it
be a delta operator $\mathcal{Q}$. The standard forward, backward and symmetric
derivatives provide typical examples of discrete delta operators; however,
one can construct infinitely many alternatives. Next, we define the Galois
differential algebra in which the delta operator $\mathcal{Q}$ acts as
a derivation, by constructing the corresponding $*$-product. The original
dynamical system is then defined in the Galois algebra, where it typically
assumes the form of a recurrence equation. The resulting discretization
procedure is inherently structure-preserving.

Thus, in general, we associate with the family of dynamical systems defined
by eq.~\eqref{lincont} a category of linear equations
$\mathcal{K}_{\{ a_{0},\ldots , a_{N},g \}}$; similarly, the family  \eqref{nonlincont} corresponds to a category of nonlinear equations, denoted
by $\mathcal{N}_{\{m; c_{0},\ldots , c_{N}\}}$. We also establish the existence
of two natural functors
$F: \mathcal{R}(\mathcal{F}) \longrightarrow \mathcal{K}_{\{ a_{0},
\ldots , a_{N},g \}}$ and
$G: \mathcal{R}(\mathcal{F}) \longrightarrow \mathcal{N}_{\{ m; c_{0},
\ldots , c_{N} \}}$. Within this algebraic framework, different discretizations
of a given continuous dynamical system are represented as morphisms within
the same category. In particular, differential equations can be mapped
into difference equations
\textit{while preserving the underlying differential structure}. Consequently,
the integrability properties of a given continuous dynamical system are
naturally inherited by its associated discrete counterparts and in particular,
solutions (regular or singular) of the continuous system are isomorphically
mapped to (regular or singular) solutions of the corresponding discrete
equations.

A crucial feature of the integrable maps\footnote{Throughout this work, the expressions ``map'',``difference equation'' or ``discrete equation'' will be used as synonyms.} associated with eqs.~\eqref{lincont} and \eqref{nonlincont}, according to the procedure outlined above, is that they can be interpreted as \textit{discrete analogues of integro-differential equations}. Indeed, these maps are nonlocal recurrences possessing the general form
\begin{equation}
P(\Delta ) u(n)=\mathcal{I}[u(n)]
\label{nonlocal}.
\end{equation}
Here $P(\Delta )$ is a polynomial in a delta operator (typically representing
a discrete derivative of a given order), $u(n)$ is the dependent variable
of the map, $n$ is the independent discrete variable, and
$\mathcal{I}[u(n)]$ is a functional depending on all values of $u$ ranging
from an initial point up to $n$. This aspect reflects the inherent nonlocality
of the equation \eqref{nonlocal}, which, in turn, originates from the intrinsic
nonlocal nature of the $*$-product. This property can be interpreted in
light of a fundamental no-go theorem proved in \cite{KSS}, which states
that it is impossible to define a field theory on an infinite lattice endowed
with a nontrivial product rule that simultaneously satisfies the Leibniz
rule, translational invariance, and locality. Consequently, the structure-preserving
discretization scheme proposed in this work - based on enforcing the Leibniz
rule for difference operators - is intrinsically linked to the nonlocal character of the resulting integrable maps.

In this scenario, the most elementary instance arises in the discretization
of linear differential equations with constant coefficients. Here, the
resulting discrete equations are standard, i.e., \textit{local} difference
equations, depending on a finite number of lattice sites determined solely
by the order of the discrete derivative. Conversely, for dynamical systems
with variable coefficients or nonlinear dependencies, the nonlocal aspects
of the theory become unavoidable and explicitly manifest.

\subsection{Main results}
\label{sec1.5}

The general discretization theory we propose relies on the following new
results.

\medskip

I) We prove that the integrable maps representing discrete versions of
the dynamical systems \eqref{lincont} or \eqref{nonlincont} satisfy the
following property: if $\sum _{k=0}^{\infty} \zeta _{k} x^{k}$ (or
$ \sum _{k=1}^{s}\frac{\zeta _{k}}{x^{k}}$) is an analytic (singular) solution
of an equation of the form \eqref{lincont} (Theorem~\ref{main1} and Theorem~\ref{main3}) or \eqref{nonlincont} (Theorem~\ref{main2} and Theorem~\ref{main4}), then $\sum _{k=0}^{n} \zeta _{k} p_{k}(n)$ (or
$\sum _{k=-s}^{-1} \zeta _{k} p_{k}(n)$), where $p_{k}(n)$ are the Laurent
basic polynomials associated with the chosen  operator, for each
$n$ is a solution of the corresponding discrete map. In other words, under
mild assumptions,
\textit{analytic (singular) solutions of the continuous systems \eqref{lincont} or \eqref{nonlincont} are transformed into regular (singular)
exact solutions of the associated integrable maps} \eqref{nonlocal}. Besides,
the discrete maps (eqs. \eqref{eqlin2}, \eqref{eqnonlin} for the regular
case and eqs. \eqref{eq:8.4} and \eqref{eq:8.14} for the singular case)
associated with the original continuous systems are derived, through our
categorical approach, in a constructive and explicit way.

\medskip
II) A Galois theory for the integrable maps considered in this work is
developed, starting from the novel notion of Rota category. Within this
setting, we also establish the existence of an isomorphism between the
\textit{Picard--Vessiot group} of a linear differential equation with constant
coefficients and that of the corresponding difference equation obtained
through our categorical approach.

\medskip

In the continuum limit, when the lattice step $h$ goes to zero, the site
$n$ goes to infinite and the product $nh$ remains constant, these new integrable
maps reduce to the original differential equations (as we will show in
several examples).

\subsection{Discussion}
\label{sec1.6}

We recall that the time scale calculus, initiated by Hilger in
\cite{Hilger1990} represents a well-established and intriguing approach
aimed at unifying discrete and continuous calculus. This unification is
achieved by defining a general domain that can be continuous, discrete,
or mixed (corresponding to time scales or, more generally, measure chains)
and by introducing appropriate jump operators on this domain. Over the
past two decades, this theory has been extensively investigated by many
authors (see, e.g., the monograph \cite{BG2017}), resulting in a substantial
body of results and applications.

However, the theoretical framework proposed in this article is independent
of, and fundamentally distinct from, time scale calculus. Instead, our
approach is based on the guiding principle of preserving the Leibniz rule
in discrete calculus -- through the notion of Rota's Galois algebra -- and
on the use of an adapted point mesh for discretization, defined by the
zeros of a set of basic polynomials.

We emphasize that our categorical approach, in principle, allows for the reformulation of the main theorems of this work for
\textit{arbitrary discrete derivatives}. This can be done by analogy with
the case of the forward difference operator, which we have analyzed explicitly
due to its prominence in applications. In particular, the extension of
our theory to discretizations based on the backward difference operator
is entirely straightforward.

\section{Algebraic preliminaries. Delta operators}
\label{sec2}

The theory of delta operators, introduced in \cite{Rota} as a foundational
framework for combinatorics, has been extensively developed in the literature
(see also \cite{RR}, \cite{Roman}). Let $\mathcal{P}$ denote the space
of polynomials in a variable $x\in \mathbb{K}$, where $\mathbb{K}$ is a
field of characteristic zero and let $\mathbb{N}$ denote the set of non--negative
integers.

Let $T$ denote the shift operator, whose action on a function $f$ is given
by $T f\left (x\right )=f\left (x+1\right )$.

\begin{definition}
\label{deltaop}
An operator $S$ is said to be \textit{shift--invariant} if it commutes with
the shift operator $T$. A shift--invariant operator $\mathcal{Q}$ is a
\textit{delta operator} if $\mathcal{Q} x=const\neq 0$.
\end{definition}
We deduce immediately the following property.
\begin{corollary}
\label{cor1}
For every constant $c\in \mathbb{R}$, $\mathcal{Q} c=0$.
\end{corollary}
The most common examples of delta operators are provided by the derivative
$D$, the forward discrete derivative $\Delta =T-1$, the backward derivative
$\nabla :=1-T^{-1}$ and the symmetric operator
$\Delta ^{s}=\frac{T-T^{-1}}{2}$.

Given a delta operator $\mathcal{Q}$, a polynomial sequence
$\{p_{k}\left (x\right )\}_{k\in \mathbb{N}}$ is said to be the sequence
of \textit{basic polynomials} for $\mathcal{Q}$ if the following conditions
are satisfied: $ (1)\text{ \ } p_{0}\left ( x\right ) =1$,
$ (2)\text{ \ }p_{k}\left ( 0\right ) =0\ \text{for all }k>0;\text{ \ }$
$ (3)\text{ \ }\mathcal{Q} p_{k}\left ( x\right ) =kp_{k-1}\left ( x
\right )$.

Notice that, for a given delta operator $\mathcal{Q}$ the sequence of associated
basic polynomials is unique.

\subsection{Formal groups and delta operators}
\label{sec2.1}

In this section, as a new result, we will demonstrate that the theory of
delta operators admits a natural interpretation within the framework of
formal group theory \cite{Ray}. Over the past decades, formal group theory
has been extensively investigated due to its central role in algebraic
topology \cite{BMN}, \cite{Lazard}, cobordism theory \cite{Quillen}, analytic
number theory \cite{Honda}, \cite{PT2010}, \cite{TempestaTAMS}, and related
fields.

Following \cite{Haze}, \cite{BMN}, we remind that a commutative one--dimensional
formal group law over a commutative, unital ring $A$ is a formal power
series~$\Phi (x,y)\in A\llbracket x,y\rrbracket $ such that
\begin{enumerate}
\item  $\Phi (x,0)=\Phi (0,x)=x$,
\item 
$\Phi \left (\Phi (x,y),z\right )=\Phi \left (x,\Phi (y,z)\right )$.
\end{enumerate}
 When $\Phi (x,y)=\Phi (y,x)$, the formal group law is said to be commutative.

Let us consider the polynomial ring
$\mathbb{Q}\left [ c_{1},c_{2},...\right ] $ and the formal group logarithm
$ F\left ( u\right ) =u+c_{1}\frac{u^{2}}{2}+c_{2}\frac{u^{3}}{3}+
\cdots $ Let $G\left ( v\right ) $ be its inverse series, i.e., the formal
group exponential
\begin{equation}
G\left ( v\right ) =v-c_{1}\frac{v^{2}}{2}+\left ( 3c_{1}^{2}-2c_{2}
\right ) \frac{v^{3}}{6}+\cdots
\label{Ap12}
\end{equation}
so that \textit{$F\left ( G\left ( v\right ) \right ) =v$}. The formal group
law associated with these series, known as the
\textit{Lazard Universal Formal Group}, is given by
\begin{equation*}
\Phi \left ( u_{1},u_{2}\right ) =G\left ( F\left ( u_{1}\right ) +F
\left ( u_{2}\right ) \right ).
\end{equation*}
It is defined over the Lazard ring $L$, i.e. the subring of
$\mathbb{Q}\left [ c_{1},c_{2},...\right ] $ generated by the coefficients
of the power series
$G\left ( F\left ( u_{1}\right ) +F\left ( u_{2}\right ) \right ) $. For
any commutative one-dimensional formal group law over any ring $A$, there
exists a unique homomorphism $L\to A$ under which the Lazard group law
is mapped into the given group law (\textit{universal property} of the Lazard
group).

In algebraic topology, delta operators are related to Thom classes and
complex cobordism theory \cite{Ray2}. In order to relate formal groups
to difference delta operators, we first describe a simple technique for
generating a family of difference delta operators. More precisely, we introduce
the difference operators
\begin{equation}
\label{2.9}
\Delta _{p}=\frac{1}{h}\sum _{k=l}^{m}\alpha _{k}T^{k}\text{,}\quad l
\text{, }m\in \mathbb{Z}\text{,}\mathbb{\quad}l<m\text{,\quad}m-l=p
\text{,} \quad \sum _{k=l}^{m}\alpha _{k}=0\text{,}\quad \sum _{k=l}^{m}k
\alpha _{k}=c
\end{equation}
where $h$ can be interpreted as a mesh spacing and $\alpha _{k}$  are
constants with $\alpha _{m}\neq 0$, $\alpha _{l}\neq 0$. We choose
$c=1$, to reproduce possibly the derivative $D$ in the continuum limit.
A difference operator of the form \eqref{2.9} is said to be a delta operator of order $p$, if it approximates the continuous derivative up to terms
of order $h^{p}$ \cite{LTW1}. Apart from the two constraints in relation \eqref{2.9}, one can arbitrarily choose $m-l-1$ more linear conditions
to fix all constants $\alpha _{k}$. Consequently, by means of the representation
$T\sim e^{v}$ (the ``symbol'' of the translation operator), each delta
operator of the family \eqref{2.9} is associated with a group exponential  \eqref{Ap12} and then with a realization of the one-dimensional Lazard
universal formal group law. Conversely, with the identification
$v\sim D$, there corresponds to each formal group exponential a delta operator
(see also \cite{Ray}, \cite{TempestaTAMS}).

\section{Laurent basic sequences of polynomials}
\label{sec3}

\subsection{Definitions}
\label{sec3.1}

In this section, we develop a novel approach to finite operator theory,
by generalizing the classical notion of basic polynomials
\cite{Rota,Roman} to the case of Laurent polynomials. In the literature,
several extensions of umbral calculus (the old denomination of finite operator
theory) are available. In particular, in \cite{Roman}, nonclassical umbral
calculi are also discussed. Also, in \cite{LR1995} and \cite{BBN1986},
several extensions to generalized finite differences calculi are reviewed.
However, the notions of Laurent basic sequences, dual sequences, twisted
Rota algebras, etc. to the best of our knowledge were not proposed before
in the literature.
\begin{definition}
\label{defn2}
Given a delta operator $\mathcal{Q}$, a Laurent basic sequence is a sequence
of rational functions $\{p_{k}(x)\}_{k\in \mathbb{Z}}$, which reduce to
polynomials for $k\in \mathbb{N}$, defined by the following properties:
\begin{align}
& 1)\text{ \ }p_{0}\left ( x\right ) =1
\notag
\\
& 2)\text{ \ }p_{k}\left ( 0\right ) =0 & k\in \mathbb{N}\setminus \{0
\}
\nonumber
\\
& 3)\text{ \ }\mathcal{Q} p_{k}\left ( x\right ) =kp_{k-1}\left ( x
\right ), & k\in \mathbb{Z} .
\nonumber
\end{align}
\end{definition}
Thus, a Laurent basic sequence for a delta operator $\mathcal{Q}$ is a
sequence of basic polynomials extended for $k<0$ to a sequence of rational
polynomials satisfying a natural ``derivation'' property.

Let us denote by $\mathcal{P}^{L}$ the space of Laurent polynomials in
a variable $x\in \mathbb{K}$. Let $\mathcal Q$ be a delta operator acting
on $\mathcal{P}^{L}$, let $\mathcal{D}$ denote the set of all delta operators,
and $\{p_{k}(x)\}_{k\in \mathbb{Z}}$ be the basic sequence of Laurent polynomials
of order $k$ uniquely associated with $\mathcal{Q}$. Let
$\mathcal{H}$ denote the algebra of formal Laurent series in $x$. Since
the polynomials $\{p_{k}(x)\}_{k\in \mathbb{Z}}$ for any choice of
$\mathcal Q$ form a basis of $\mathcal{H}$, any $f\in \mathcal{H}$ can
be expanded as a formal Laurent series of the form
$ f(x)=\sum _{k=-\infty}^{\infty}a_{k} p_{k}(x)
\label{exp}
 $. In this way, we can extend the action of delta operators on functions.
Let $\mathcal{L}$ be a set of points on the real line, isomorphic to
$\mathbb{Z}$. Denote by $\mathcal{H_{L}}$ the vector space of the formal
power series defined on $\mathcal{L}$. The space $\mathcal{H}$ (and, consequently,
$\mathcal{H_{L}}$) can be endowed with an algebra structure by introducing
a new, suitable product.

\begin{definition}
\label{defn3}
Given a delta operator $\mathcal{Q}$ and the associated Laurent basic sequence
$\{p_{k}(x)\}_{k\in \mathbb{Z}}$, the $*_{\mathcal{Q}}$ product is defined
via the relation
\begin{equation}
p_{k}(x)*_{\mathcal{Q}} p_{j}(x):=p_{k+j}(x), \qquad k,j \in
\mathbb{Z}.
\label{starproduct}
\end{equation}
\end{definition}
\begin{remark}
\label{rem1}
An analogous product, for the case of the standard basic polynomial sequence
associated with the forward difference operator $\Delta $, was proposed
in \cite{Ward} and in \cite{ismail}. In what follows, we will use the symbol
$*$ whenever the choice of the delta operator $\mathcal{Q}$ will be obvious.
It can be shown that the space
$(\mathcal{H}, +,\cdot , *_{\mathcal{Q}})$, equipped with the usual operations
of series addition, scalar multiplication, and the $*$-product  \eqref{starproduct}, forms an associative algebra.
\end{remark}

\subsection{The natural Laurent sequence}
\label{sec3.2}

The simplest example of a Laurent basic sequence is given by
\begin{equation}
\label{nL}
p_{k}(x):= x^{k}, \qquad \qquad k\in \mathbb{Z}
\end{equation}
for the delta operator $\mathcal{Q}= D$. We shall call it the
\textit{natural Laurent basic sequence}. In order to construct other nontrivial
examples, we will introduce the notion of dual sequences of basic polynomials.

\subsection{Dual sequences}
\label{sec3.3}

We denote by $p^{+}_{k}(x):=x(x-1)\cdot \ldots \cdot (x-k+1)$,
$k\in \mathbb{N}$, the sequence of basic polynomials for the operator
$\Delta =T-1$. Similarly,
$p^{-}_{k}(x):=x(x+1)\cdot \ldots \cdot (x+k-1)$ is the basic sequence
for $\nabla =1-T^{-1}$.

We shall introduce sequences of \textit{rational basic polynomials}. To
this aim, let us observe that
\begin{eqnarray*}
\Delta \frac{1}{p^{-}_{k}(x)}&=&\frac{1}{p^{-}_{k}(x+1)}-
\frac{1}{p^{-}_{k}(x)} =-\frac{k}{p^{-}_{k+1}(x)}.
\end{eqnarray*}
Thus, if we define:
\begin{equation}
q^{+}_{k}(x):=\frac{1}{p^{-}_{k}(x)}, \quad k\in \mathbb{N}\setminus\{0\},
\qquad q^{+}_{0}=1
\label{eq3.3}
\end{equation}
we obtain
\begin{equation}
\Delta q^{+}_{k}(x)=-k\,q^{+}_{k+1}(x).
\label{eq3.4}
\end{equation}
Analogously, for the operator $\nabla $ we have
\begin{eqnarray*}
\nabla \frac{1}{p^{+}_{k}(x)}&=&\frac{1}{p^{+}_{k}(x)}-
\frac{1}{p^{+}_{k}(x-1)}=-\frac{k}{p^{+}_{k+1}(x)}.
\end{eqnarray*}
Consequently, denoting by
\begin{equation}
q^{-}_{k}(x)=\frac{1}{p^{+}_{k}(x)}, \quad k\in \mathbb{N}\setminus
\{0\}, \qquad q^{-}_{0}=1
\label{eq3.5}
\end{equation}
the rational functions associated with the basic polynomials for
$\Delta $, we obtain that
\begin{equation}
\nabla q^{-}_{k}(x)=-k\,q^{-}_{k+1}(x).
\label{eq3.6}
\end{equation}

\begin{definition}
\label{defn4}
The sequences $\{p^{+}_{k}(x)\}_{k\in \mathbb{N}}$ and
$\{q^{+}_{k}(x)\}_{k\in \mathbb{N}}$ are said to be dual for the operator
$\Delta $; similarly, $\{p^{-}_{k}(x)\}_{k\in \mathbb{N}}$ and
$\{q^{-}_{k}(x)\}_{k\in \mathbb{N}}$ are dual sequences for
$\nabla $.
\end{definition}
 As a consequence of the previous discussion, we have proved the following
\begin{proposition}
\label{prop1}
The sequence
\begin{equation}
\label{eq:LSDelta}
p_{k}(x):=
\begin{cases}
x(x-1)\cdots (x-k+1) &\qquad k\in \mathbb{N}\setminus \{0\},
\\
\dfrac{1}{x(x+1)\cdots (x-k-1)} &\qquad k \in \mathbb{Z}^{-},
\\
1 &\qquad k=0,
\end{cases}
\end{equation}
is a Laurent basic sequence for the delta operator $\Delta $.
\end{proposition}
\textbf{Open Problem}. It would be very interesting to ascertain whether
there exist infinitely many dual sequences.

\subsection{The Abel-Laurent basic sequence}
\label{sec3.4}

Consider the Abel delta operator:
\begin{equation}
\mathcal{Q}^{A}=DT^{a}
\label{eq3.8}
\end{equation}
with $D=\frac{\mathrm{d}}{\mathrm{d}x}$, $T^a f(x)=f(x+a)$, $a>0$.
The Laurent basic sequence for the Abel operator is given by
\begin{equation}
p^{A}_{k}(x):=
\begin{cases}
x(x-ak)^{k-1} &\qquad k\in \mathbb{N}\setminus \{0\},
\\
(x-ak)^{k} &\qquad k \in \mathbb{Z}^{-}\cup\{0\}.
\end{cases}
\label{eq3.9}
\end{equation}

\subsection{The symmetric Laurent basic sequence}
\label{sec3.5}

The symmetric delta operator $\Delta ^{s}:= \frac{T-T^{-1}}{2}$ is also
relevant in several applicative contexts. The basic Laurent polynomials
associated with $\Delta ^{s}$ are
\begin{equation}
p^{s}_{k}(x):=
\begin{cases}
x \prod_{j=1}^{k-1} (x + k - 2j) &\qquad k\in \mathbb{N}\setminus \{0\},
\\
1 &\qquad k = 0,
\\
\dfrac{1}{\prod_{j=0}^{-k-1} (x - k - 1 - 2j)} &\qquad k\in \mathbb{Z}^{-}.
\end{cases}
\label{eq3.10}
\end{equation}

\subsection{Lattices and Laurent basic sequences}
\label{sec3.6}

Consider the uniform lattice in $\mathbb{R}^{+}\cup \{0\}$ with sites at
the integers and the Laurent sequence  \eqref{eq:LSDelta}. Then the values
of $p_{k}(x)$ at the sites $x=n\in \mathbb{N}$, with
$k\in \mathbb{Z}$, explicitly read
\begin{eqnarray*}
p_{k}(n)=
\begin{cases}
\prod _{j=0}^{k-1}(n-j)=\dfrac{n!}{(n-k)!} & \qquad k\ge 0,
\\[15pt]
\dfrac{1}{\prod _{j=0}^{-k-1}(n+j)}=\dfrac{(n-1)!}{(n-1-k)!} &
\qquad k< 0.
\end{cases}
\end{eqnarray*}
Hereafter, we shall assume that negative factorials in the denominator
make the fraction vanish.

\section{Category theory and dynamical systems}
\label{sec4}

\subsection{Twisted Rota algebras}
\label{sec4.1}

The notion of Rota algebra has been introduced in \cite{TempestaJDE}, as
the natural Galois differential algebra where the discretization procedure
is carried out. In this section, this notion will be generalized to the
case of Laurent formal power series. We shall also develop a categorical
approach.

\begin{definition}
\label{defn5}
A twisted Rota differential algebra is a Galois differential algebra
$(\mathcal{H}, \mathcal{Q})$, where
$(\mathcal{H}, +, \cdot , *_{\mathcal{Q}})$ is an associative algebra of
Laurent formal power series, the product $*_{\mathcal{Q}}$ is the composition
law defined by \eqref{starproduct}, and $\mathcal{Q}$ is a delta operator
acting as a derivation on $\mathcal{H}$:
\begin{equation}
i) \quad \mathcal{Q}(a+b)=\mathcal{Q}(a)+\mathcal{Q}(b),
\qquad
\mathcal{Q}(\lambda a)=\lambda \mathcal{Q}(a), \quad \lambda \in
\mathbb{K},
\label{eq4.1}
\end{equation}
\begin{equation}
ii) \quad \mathcal{Q}(a*b)= \mathcal{Q}(a)*b+a*\mathcal{Q}(b).
\label{eq4.2}
\end{equation}
\end{definition}

\subsection{The generalized Rota category}
\label{sec4.2}

\begin{definition}
\label{defn6}
The generalized Rota category, denoted by $\mathcal{R}(\mathcal{H})$, is
the collection of all Rota algebras $(\mathcal{H}, \mathcal{Q})$, with
morphisms defined by
\begin{equation}
\rho _{\mathcal{Q},\mathcal{Q}'}: (\mathcal{H}, +, \cdot , *_{
\mathcal Q})\longrightarrow (\mathcal{H}, +, \cdot , *_{\mathcal{Q'}})
\label{rotamorphism}
\end{equation}
which are closed under composition.
\end{definition}
The action of the morphism $\rho _{\mathcal{Q},\mathcal{Q}'}$ on formal
power series is defined by
\begin{equation}
\sum _{k} \zeta _{k} p_{k} (x) \longrightarrow \sum _{m} \zeta _{m} q_{m}(x),
\label{morphism}
\end{equation}
where $\{p_{k}(x)\}_{k\in \mathbb{N}}$ and
$\{q_{m}(x)\}_{m\in \mathbb{N}}$ are the basic sequences associated with
$\mathcal{Q}$ and $\mathcal{Q}'$ respectively. The property of closure
under composition is trivial.

\subsection{New categories for dynamical systems}
\label{sec4.3}

We introduce two new definitions of categories for differential equations
of the form \eqref{lincont} and  \eqref{nonlincont}. For
$N\in \mathbb{N}\setminus \{0\}$, define
$u^{*N}:=\underbrace{u*\ldots *u}_{N-times}$.
\begin{definition}
\label{defn7}
For any choice of the set of functions
$\{a_{0}(x),\ldots ,a_{N}(x), g(x)\}$, the category
$\mathcal{K}_{\{ a_{0},\ldots , a_{N},g\}}$ of linear dynamical systems
of order $N$ is the collection of all equations of the form
\begin{equation}
lin(\mathcal{Q},u,*_{\mathcal{Q}}):=a_{N}(x)* \mathcal{Q}^{N} y + a_{N-1}(x)*
\mathcal{Q}^{(N-1)}y +\ldots + a_{1}(x)* \mathcal{Q} y+a_{0}(x)* y +g(x)=0.
\label{linabstr}
\end{equation}
 The set of correspondences
\begin{equation}
\lambda _{\mathcal{Q},\mathcal{Q}'}: \mathcal{K}_{\{ a_{0},\ldots , a_{N},g
\}} \longrightarrow \mathcal{K}_{\{ a_{0},\ldots , a_{N},g \}},
\label{eq4.6}
\end{equation}
\begin{equation}
lin(\mathcal{Q},y,*) \longrightarrow lin(\mathcal{Q}',y,*')
\label{morfeq},
\end{equation}
defines the class of morphisms of the category.
\end{definition}
In other words, given a differential equation of the form  \eqref{lincont}, one may define a category whose objects are the equation  \eqref{lincont} together with all infinitely many of its discretizations,
obtained by varying $\mathcal{Q}\in \mathcal{D}$ and constructing the corresponding
morphisms  \eqref{morfeq}.
\begin{definition}
\label{defn8}
For any choice of the set of functions
$\{c_{0}(x),\ldots ,c_{N}(x)\}$, the category
$\mathcal{N}_{\{m; c_{0},\ldots , c_{N}\}}$ of nonlinear dynamical systems
of order $m\in \mathbb{N}$ consists of all equations of the form
\begin{equation}
eq(\mathcal{Q},y,*_{\mathcal{Q}}):= \mathcal{Q}^{m} y- c_{N}(x)*y^{*N}-c_{N-1}(x)*y^{*(N-1)}-
\ldots -c_{1}(x)*y-c_{0}(x)=0.
\label{nonlinabstr}
\end{equation}
The set of correspondences
\begin{equation}
\nu _{\mathcal{Q},\mathcal{Q}'}: \mathcal{N}_{\{m; c_{0},\ldots , c_{N}
\}} \longrightarrow \mathcal{N}_{\{m; c_{0},\ldots , c_{N}\}},
\label{eq4.9}
\end{equation}
\begin{equation}
eq(\mathcal{Q},y,*) \longrightarrow eq(\mathcal{Q}',y,*')
\label{morfeqnlin},
\end{equation}
defines the class of morphisms of the category.
\end{definition}
 The closure of the morphisms $\lambda _{\mathcal{Q},\mathcal{Q}'}$ and
$\nu _{\mathcal{Q},\mathcal{Q}'}$ under composition is easily verified.
\begin{definition}
\label{defn9}
Two objects in the category
$\mathcal{K}_{\{ a_{0},\ldots , a_{N},k\}}$ (or in
$\mathcal{N}_{\{m; c_{0},\ldots , c_{N}\}}$) are said to represent two
categorically equivalent equations. Alternatively, two equations are categorically
equivalent if there exists a morphism of categories
$\lambda _{\mathcal{Q},\mathcal{Q}'}$ (or
$\nu _{\mathcal{Q},\mathcal{Q}'}$) mapping one equation to the other.
\end{definition}
A notable property of the categories introduced above is their hierarchical
structure via subcategories, with inclusions determined by the natural
identifications:
\begin{eqnarray*}
\mathcal{K}_{\{ a_{0},\ldots , a_{k}\}}:= \mathcal{K}_{\{ a_{0},
\ldots , a_{k},\underbrace{0,\ldots ,0}_{(N-k)-\mathrm{times}}\}},
\\
\mathcal{N}_{\{m; c_{0},\ldots , c_{k}\}}:= \mathcal{N}_{\{ m; c_{0},
\ldots , c_{k},\underbrace{0,\ldots ,0}_{(N-k)-\mathrm{times}}\}},
\end{eqnarray*}
for $k<N$.
\begin{lemma}
\label{lem1}
For any fixed $N\in \mathbb{N}$, and for any choice of the polynomials
$\{ a_{0},\ldots , a_{N},g\}$, there exists a filtration of subcategories,
given by the finite sequences
\begin{equation}
\mathcal{K}_{\{ a_{0}\}}\subset \mathcal{K}_{\{ a_{0},a_{1}\}}
\subset \ldots \subset \mathcal{K}_{\{ a_{0},\ldots , a_{k}}\}
\subset \ldots \subset \mathcal{K}_{\{ a_{0},\ldots , a_{N},g\}}
\label{filtrF}
\end{equation}
and
\begin{equation}
\mathcal{N}_{\{m; c_{0}\}}\subset \mathcal{N}_{\{m; c_{0},c_{1}\}}
\subset \ldots \subset \mathcal{N}_{\{m; c_{0},\ldots , c_{k}}\}
\subset \ldots \subset \mathcal{N}_{\{m; c_{0},\ldots , c_{N}\}},
\label{filtrG}
\end{equation}
respectively.
\end{lemma}
\begin{proof}
For each of the subsets $\mathcal{K}_{\{ a_{0},\ldots , a_{k}}\}$, the
restriction of the morphisms $\lambda _{\mathcal{Q},\mathcal{Q}'}$ over
these subsets still preserves the closure under composition and defines
morphisms of subcategories. The same argument holds for the sequence related
to $\mathcal{N}_{\{m;c_{0},\ldots , c_{N}\}}$.
\end{proof}
As a consequence of the previous construction, one can define functors
between the category of Rota differential algebras and the categories of
abstract dynamical systems defined above.
\begin{theorem}
\label{functor}
For any choice of the functions
$\{ a_{0}(x),\ldots , a_{N}(x),g(x)\}$, the application
\begin{equation}
F: \mathcal{R}(\mathcal{H}) \longrightarrow \mathcal{K}_{\{ a_{0},
\ldots , a_{N},g \}},
\label{functorF}
\end{equation}
\begin{equation*}
(\mathcal{H},+, \cdot , *_{\mathcal{Q}}) \longrightarrow lin(
\mathcal{Q},y,*),
\end{equation*}
\begin{equation*}
\rho _{\mathcal{Q},\mathcal{Q}'} \longrightarrow \lambda _{
\mathcal{Q},\mathcal{Q}'},
\end{equation*}
is a covariant functor.
\end{theorem}
\begin{proof}
A direct verification shows that $F$ preserves the composition of morphisms:
\begin{equation*}
F(\rho _{\mathcal{Q}'',\mathcal{Q}'}\circ \rho _{\mathcal{Q}',
\mathcal{Q}})=F(\rho _{\mathcal{Q}'',\mathcal{Q}'})\circ F(\rho _{
\mathcal{Q}',\mathcal{Q}}).
\end{equation*}
If we denote by $id_{\mathcal{Q}}:=\rho _{\mathcal{Q},\mathcal{Q}}$ the
identity morphism, we also have
\begin{equation*}
F(id_{\mathcal{Q}}(\mathcal{A}))=id_{\mathcal{Q}}(F(\mathcal{A})),
\end{equation*}
where $\mathcal{A}\in \mathcal{R}(\mathcal{F})$.
\end{proof}
 In the same manner one can prove that the application
\begin{equation}
G: \mathcal{R}(\mathcal{H}) \longrightarrow \mathcal{N}_{\{ m; c_{0},
\ldots , c_{N} \}},
\label{functorG}
\end{equation}
\begin{equation*}
(\mathcal{H},+, \cdot , *_{\mathcal{Q}}) \longrightarrow eq(
\mathcal{Q},y,*),
\end{equation*}
\begin{equation*}
\rho _{\mathcal{Q},\mathcal{Q}'} \longrightarrow \nu _{\mathcal{Q},
\mathcal{Q}'},
\end{equation*}
is a covariant functor.

The functors  \eqref{functorF} and  \eqref{functorG} capture the essential
aspects of the proposed discretization, and provide
\textit{functorial Rota correspondences} between continuous and discrete
dynamical systems. Moreover, these functors can be naturally extended to
the subcategories introduced above.
\begin{corollary}
\label{cor2}
The restriction of the functors $F$ and $G$ to the subcategories defined
by the filtrations  \eqref{filtrF} and  \eqref{filtrG} respectively remain
covariant functors.
\end{corollary}

\section{Main theorems for the regular case}
\label{sec:5}

In this section, we present the main results of the article concerning
the discretization of ODEs admitting analytical solutions. Precisely, using
the new categorical framework developed above, we construct a class of
integrable maps associated with the dynamical systems \eqref{lincont} and  \eqref{nonlincont}, respectively. These maps are categorically equivalent
to the original continuous systems and represent them in the Galois algebra
corresponding to the forward difference operator. The discretization is
performed on an \emph{equally spaced} mesh.

Theorem~\ref{main1} establishes a solution to the problem of the integrability-preserving
discretization of linear $n$-th order ODEs with analytic variable coefficients.
 Theorem~\ref{main2} solves the analogous problem for the case of nonlinear
first-order ODEs.

\begin{remark}
In both Theorems~\ref{main1} and \ref{main2}, we will define our maps in
the Rota algebra $(\mathcal{F}, \mathcal{Q})$, where $\mathcal{F}$ is the
space of formal power series, and $p_{k}(x)_{k\in \mathbb{N}}$ are basic
sequences for the delta operator $Q=\Delta $. In section~\ref{sec:8}, when
dealing with singular solutions, this picture will be extended to the case
of Laurent formal series, and Laurent basic polynomials.
\end{remark}

\subsection{Integrable maps from linear ODEs}
\label{sec5.1}

\begin{theorem}
\label{main1}
 Consider the differential equation
\begin{equation}
lin(\partial , y):=a_{N}(x) \frac{d^{N}}{dx^{N}}y + a_{N-1}(x)
\frac{d^{N-1}}{dx^{N-1}}y +\ldots +a_{1}(x) \frac{d}{dx} y+a_{0}(x) y +g(x)=0,
\label{lincont2}
\end{equation}
where $a_{i}(x)$, $i=0,\ldots ,N$ and $g(x)$ are real analytic functions at
$x=0$ (with $a_{N}(0)\neq 0)$:
\begin{equation}
a_{i}(x)=\sum _{k_{i}=0}^{\infty}\frac{\alpha _{ik_{i}}}{k_{i}!}x^{k_{i}},
\qquad g(x)=\sum _{\ell =0}^{\infty}\frac{\gamma _{\ell}}{\ell !}x^{
\ell}.
\label{eq5.2}
\end{equation}
Assume that
\begin{equation}
\label{eq:23}
y(x)=\sum _{k=0}^{\infty} \zeta _{k} x^{k}
\end{equation}
be a real solution of  \eqref{lincont2}, analytic at $x=0$. Then the difference
equation
\begin{equation}
\label{eqlin}
\sum _{i=0}^{N} \sum _{k_{i}=0}^{n}
\frac{n!\alpha _{ik_{i}}}{(n-k_{i})!k_{i}!}\Delta ^{i}u_{n-k_{i}} +
\sum _{\ell =0}^{n} \frac{n! \gamma _{\ell}} {\ell ! (n-\ell )!} =0,
\end{equation}
or, in the equivalent form, the equation
\begin{equation}
\label{eqlin2}
\sum _{i=0}^{N} \sum _{k_{i}=0}^{n} \sum _{k=0}^{i} (-1)^{i-k}
\binom{n}{k_{i}}\binom{i}{k}\alpha _{ik_{i}}u_{n-k_{i}+k} +\sum _{
\ell =0}^{n} \frac{n! \gamma _{\ell}} {\ell ! (n-\ell )!} =0,
\end{equation}
defined on a regular set of points $\mathcal{L}\subset \mathbb{R}$ indexed
by the variable $n\in \mathbb{N}$, admits as a solution the series
\begin{equation}
u_{n}=\sum _{k=0}^{n} \zeta _{k} \frac{n!}{(n-k)!}.
\label{part1}
\end{equation}
\end{theorem}

\begin{proof}
We start by applying the morphism  \eqref{morphism} to the series  \eqref{eq:23}; we obtain the transformation
\begin{equation}
\sum _{k=0}^{\infty} \zeta _{k} x^{k} \longrightarrow \sum _{k=0}^{
\infty} \zeta _{k} p_{k} (x).
\hspace{10mm}
\label{interpol}
\end{equation}
The discrete transform \eqref{interpol} is finite whenever
$x\in \mathcal{L}$, provided that the points of the mesh
$\mathcal{L}$ coincide with the zeros of the basic sequence of polynomials
$\{p_{k}(x)\}_{k\in \mathbb{N}}$. In the following, we choose our basic
sequence to be the set of lower factorial polynomials. Consequently,
$\mathcal{L}$ denotes an equally spaced set of points on the real nonnegative
axis, indexed by $n\in \mathbb{N}$. Thus, we have
\begin{equation}
p_{k}(n)=
\begin{cases}
0, & \mathrm{if} \quad n<k,
\\[5pt]
\dfrac{n!}{(n-k)!}, & \mathrm{if}\quad n \geq k.
\label{pol}
\end{cases}
\end{equation}
 Let us introduce the function $u:\mathbb{N}\rightarrow \mathbb{R}$ defined
by
\begin{equation}
\label{eq:3.9}
u_{n}= \sum _{k=0}^{n} \frac{n!}{(n-k)!} \zeta _{k}.
\end{equation}
For its \textit{inverse interpolating transform}, we have
\begin{equation}
\label{invtr}
\zeta _{k}= \sum _{l=0}^{k}\frac{ (-1)^{k-l} }{ l!(k-l)!}u_{l}.
\end{equation}
For the purpose of constructing the equation categorically equivalent to
eq.~\eqref{lincont2} in $\mathcal{K}_{\{ a_{0},\ldots , a_{N},g\}}$, we
represent the product
$x^{r}\frac{\mathrm{d}^{s}}{\mathrm{d}x^{s}}y(x)$ on the set
$\mathcal{L}$ by the action of the morphism
$\rho _{\partial ,\Delta}: (\mathcal{F},+,\cdot )\rightarrow (
\mathcal{F},+,*_{\Delta})$ between the standard (Rota) algebra of power
series endowed with the pointwise product and that equipped with the forward
difference operator. This morphism determines the general (umbral) correspondences
\begin{equation}
x^{r}\frac{\mathrm{d}^{s}}{\mathrm{d}x^{s}}y(x)\longrightarrow
\frac{n!}{(n-r)!}\Delta ^{s} u_{n-r}, \qquad r,s\in \mathbb{N}
\label{eq5.11}
\end{equation}
which can be specialized to the following cases:
\begin{align}
y(x)&\longrightarrow u_{n},\\
x^{r} y(x)&\longrightarrow \frac{n!}{(n-r )!} u_{n-r},\\
\frac{\mathrm{d}^{s}}{\mathrm{d}x^{s}}y(x)&\longrightarrow \Delta ^{s}
u_{n}.
\end{align}
To prove these relations, observe that
\begin{equation}
x^{r}\frac{\mathrm{d}^{s}}{\mathrm{d}x^{s}}y(x)\longrightarrow p_{r}(n)*
\sum _{k=0}^{n+s} \frac{k!}{(k-s)!}\,\zeta _{k}\,p_{k-s}(n)=\sum _{k=0}^{n+s-r}
\frac{k!}{(k-s)!}\,\zeta _{k}\,p_{k-s+r}(n).
\label{eq5.15}
\end{equation}
Then, we deduce that
\begin{eqnarray}
&&\sum _{k=0}^{n+s-r} \frac{k!}{(k-s)!}\,\zeta _{k}\,p_{k-s+r}(n)=
\sum _{k=0}^{n+s-r} \frac{n!\,k!}{(k-s)!(n-k+s-r)!}\zeta _{k}
\label{eq5.16}
\\
&=&\sum _{k=0}^{n+s-r}\sum _{j=0}^{k}
\frac{n!\,k!}{(k-s)!(n-r-(k-s))!}\frac{(-1)^{k-j}}{j!(k-j)!}u_{j}
\nonumber
\\
&=&\frac{n!}{(n-r)!}\sum _{j=0}^{n+s-r} \left (\sum _{k=j}^{n+s-r} (-1)^{k-j}
\binom{n-r}{k-s}\binom{k}{j}\right )u_{j}.
\nonumber
\end{eqnarray}
It can be shown that (see the Appendix~\ref{AppendixA}):
\begin{equation}
\sum _{k=j}^{n+s-r} (-1)^{k-j} \binom{n-r}{k-s}\binom{k}{j}=\sum _{i=0}^{s}(-1)^{s-i}
\binom{s}{i}\delta _{n-r+i,j}.
\label{eq5.17}
\end{equation}
Thus,
\begin{eqnarray}
&& x^{r}\frac{\mathrm{d}^{s}}{\mathrm{d}x^{s}}y(x)\longrightarrow
\frac{n!}{(n-r)!}\sum _{j=0}^{n+s-r} \left (\sum _{i=0}^{s}(-1)^{s-i}
\binom{s}{i}\delta _{n-r+i,j}\right )u_{j}
\label{eq5.18}
\\
&=&\frac{n!}{(n-r)!}\sum _{i=0}^{s}(-1)^{s-i}\binom{s}{i}u_{n-r+i}  =
\frac{n!}{(n-r)!}\Delta ^{s}u_{n-r}.
\nonumber
\end{eqnarray}

By using the previous formulas, we deduce that  eq.~\eqref{eqlin} is categorically
equivalent to eq.~\eqref{lincont2}. Both in turn are representations of
the abstract equation  \eqref{linabstr}.

To conclude the proof, observe that, by means of the action of the functor
$F$, any formal power series $y$ of eq.~\eqref{lincont2} is carried into
a solution $u_{n}$ of the equation
\begin{equation*}
lin(\mathcal{Q},u,*_{\mathcal{Q}})=\lambda _{\partial ,\Delta}(lin(
\partial ,y,\cdot )).
\end{equation*}
In addition, on the mesh $\mathcal{L}$, the sum
$\sum _{k} \zeta _{k} p_{k}(n)$ truncates and converts into the finite
sum  \eqref{part1}.
\end{proof}

\begin{remark}
\label{rem10}
Theorem~\ref{main1} can be considerably generalized by choosing an arbitrary
different delta operator $\mathcal{Q}$ and the corresponding basic sequence
$\{p_{k}(x)\}_{k\in \mathbb{N}}$. In this case, we define the lattice
$\mathcal{L}$ to be the union set of all zeroes of the polynomials of the
sequence $\{p_{k}(x)\}_{k\in \mathbb{N}}$. The action of the morphism
$\lambda _{\partial ,\mathcal{Q}}$ will provide a different, categorically
equivalent representation of eq.~\eqref{lincont2} in
$\mathcal{K}_{\{ a_{0},\ldots , a_{N},g\}}$.
\end{remark}

\subsection{Integrable maps from nonlinear ODEs}
\label{sec5.2}

\begin{theorem}
\label{main2}
 Consider a dynamical system of the form
\begin{equation}
eq(\partial , y):=\frac{d^{m}}{dx^{m}}y-\sum _{r=0}^{N} c_{r}(x) y^{r}=0,
\label{ncont2}
\end{equation}
where $N, m\in \mathbb{N}\setminus \{0\}$,
$c_{l}(x)= \sum _{k_{l}=0}^{\infty}\beta _{l k_{l}} x^{k_{l}}$ are real analytic
at $x=0$, with $\beta _{l k_{l}}\in \mathbb{R}$, $l=0,\ldots ,N$. Assume
that
\begin{equation}
y(x)=\sum _{k=0}^{\infty} \zeta _{k} x^{k}
\label{eq5.20}
\end{equation}
be a real solution of \eqref{ncont2}, analytic at $x=0$. Then the difference
equation
\begin{equation}
\sum _{i=0}^{m}(-1)^{m-i}\binom{m}{i}u_{n+i}-\sum _{r=0}^{N} \sum _{s_{r},j_{1},
\ldots ,j_{r}=0}^{n}\beta _{rs_{r}}
\frac{(1-r)^{n-s_{r}-j_{1}-\ldots -j_{r} } n!}{ (n-s_{r}-j_{1}-\ldots -j_{r})! }
\prod _{l=1}^{r}\frac{u_{j_{l}}}{j_{l}!}=0,
\label{eqnonlin}
\end{equation}
which represents eq.\eqref{ncont2} on a regular set of points
$\mathcal{L}\subset \mathbb{R}$ indexed by the variable
$n\in \mathbb{N}$, admits as a solution the series
\begin{equation}
u_{n}=\sum _{k=0}^{n} \zeta _{k} \frac{n!}{(n-k)!}.
\label{part}
\end{equation}
\end{theorem}

\begin{proof}
We consider the correspondence  
\begin{equation}
y^{(r)}(x)\to \Delta ^{r} u_{n}=\sum _{j=0}^{r}(-1)^{r-j}\binom{r}{j}u_{n+j},
\label{eq5.23}
\end{equation}
where $u_{n}$ is related to its inverse transform by the reciprocal eqs.
 \eqref{eq:3.9} and  \eqref{invtr}. We will also determine how the morphism
$\rho _{\partial ,\Delta}: (\mathcal{F},+,\cdot )\rightarrow (
\mathcal{F},+,*_{\Delta})$ acts on the powers of $y(x)$. Let
$\mathbf{k}$, $\mathbf{j}$ be the multiindices:
$\mathbf{k}=(k_{1},\ldots ,k_{r})$,
$\mathbf{j}=(j_{1},\ldots ,j_{r})$, $k=\sum _{i=1}^{r} k_{i}$,
$j=\sum _{i=1}^{r} j_{i}$ (see the Appendix~\ref{AppendixB} for a detailed
proof of eqs.  \eqref{eq:3.24} and  \eqref{eq:3.25}):
\begin{eqnarray}
\label{eq:3.24}
&& y(x)^{r}=\left (\sum _{m=0}^{\infty} \zeta _{m} x^{m}\right )^{r}=
\sum _{\mathbf{k}=0}^{\infty} \zeta _{k_{1}}\cdots \zeta _{k_{r}} x^{k}
\to \sum _{\mathbf{k} =0}^{\infty} \zeta _{k_{1}}\cdots \zeta _{k_{r}}
p_{k}(x)
\\
\nonumber
&\to & \sum _{\mathbf{k}=0}^{\infty} \frac{n!}{(n-k)!}\prod _{i=1}^{r}
\left (\sum _{j_{i}=0}^{k_{i}}
\frac{(-1)^{k_{i}-j_{i}}}{j_{i}!(k_{i}-j_{i})!}u_{j_{i}}\right ) =
\sum _{\mathbf{j}=0,\,j\le n} \frac{(1-r)^{n -j }n!}{(n-j)! }\prod _{i=1}^{r}
\frac{u_{j_{i}}}{j_{i}!}.
\end{eqnarray}
As usual, factorials of negative numbers in the denominators yield vanishing
fractions. Then, the last series in eq.\eqref{eq:3.24} have a finite number
of terms only.

Also, we observe that the products of the form $x^{s} y^{r}$ are mapped
into the terms
\begin{eqnarray}
\label{eq:3.25}
&& x^{s}y^{r}=x^{s}\left (\sum _{m=0}^{\infty}\zeta _{m} x^{m}\right )^{r}=
\sum _{\mathbf{k}=0}^{\infty} \zeta _{k_{1}}\cdots \zeta _{k_{r}} x^{k+s}
\sum _{\mathbf{k}=0}^{\infty} \zeta _{k_{1}}\cdots \zeta _{k_{r}} p_{k+s}(x)
\\
&\to & \sum _{\mathbf{k}=0}^{\infty} \frac{n!}{(n-s-k)!}\prod _{i=1}^{r}
\left (\sum _{j_{i}=0}^{k_{i}}
\frac{(-1)^{k_{i}-j_{i}}}{j_{i}!(k_{i}-j_{i})!}u_{j_{i}}\right )
\nonumber
\\
& = & \sum _{\mathbf{j}=0,\,j\le n-s}
\frac{(1-r)^{n-s -j } n!}{ (n-s-j)! }\prod _{i=1}^{r}
\frac{u_{j_{i}}}{j_{i}!}.
\nonumber
\end{eqnarray}
The product of a function
$c(x)=\sum _{s=0}^{\infty} \alpha _{s}x^{s}$ times a power of $y$ is mapped
into:
\begin{eqnarray}
c(x)y^{r}&=&\left (\sum _{s=0}^{\infty}\beta _{s} x^{s}\right )
\left (\sum _{m=0}^{\infty}\zeta _{m} x^{m}\right )^{r}=\left (\sum _{s=0}^{
\infty}\beta _{s}x^{s}\right ) \left (\sum _{\mathbf{k}=0}^{\infty}
\zeta _{k_{1}}\cdots \zeta _{k_{r}} x^{k}\right )
\nonumber
\\
&= &\sum _{s,\mathbf{k}=0}^{\infty} \beta _{s} \zeta _{k_{1}}\cdots
\zeta _{k_{r}} x^{k+s} \to \sum _{s,\mathbf{k}=0}^{\infty} \beta _{s}
\zeta _{k_{1}}\cdots \zeta _{k_{r}} p_{k+s}(x)
\nonumber
\\
&\to & \sum _{s,\mathbf{k}=0}^{\infty} \beta _{s} \zeta _{k_{1}}
\cdots \zeta _{k_{r}} \frac{n!}{(n-k-s)!} = \sum _{s,\mathbf{j}=0}^{
\infty}\beta _{s}\frac{(1-r)^{n-s -j } n!}{ (n-s-j)! }\prod _{i=1}^{r}
\frac{u_{j_{i}}}{j_{i}!}.
\label{eq5.26}
\end{eqnarray}

By combining together all the previous results, we obtain the proof that
the difference equation  \eqref{eqnonlin} is categorically equivalent to
the equation
\begin{equation}
\sum _{j=0}^{m}(-1)^{m-j}\binom{m}{j}u_{n+j}-\sum _{q=0}^{N} \sum _{s_{q},
\,\mathbf{j}=0}^{\infty}\beta _{qs_{q}}
\frac{(1-r)^{n-s_{q}-j } n!}{ (n-s_{q}-j)! }\prod _{i=1}^{q}
\frac{u_{j_{i}}}{j_{i}!}=0
\label{eq5.27}
\end{equation}
defined on $\mathcal{L}$, i.e. is the image of eq.\eqref{ncont2} under
the action of the morphism $\nu _{\partial ,\Delta}$ defined in
$\mathcal{N}_{\{m; a_{0},\ldots , a_{N}\}}$. In turn, both are realizations
of  eq.~\eqref{nonlinabstr}.

To prove that the series \eqref{part} is solution of eq.~\eqref{eqnonlin}, observe that the morphism
$\rho _{\partial ,\Delta}$ provides the correspondence
\begin{equation}
\sum _{k} \zeta _{k} x^{k} \longrightarrow \sum _{k} \zeta _{k} p_{k}(n).
\label{eq5.28}
\end{equation}

The action of the functor $G$ will carry any analytic solution $y$ of eq.~\eqref{ncont2},
defined on the algebra $(\mathcal{F}+,\cdot )$, into a solution
$u_{n}$ of the corresponding equation
\begin{equation*}
eq(\Delta ,u,*_{\Delta})=\nu _{\partial ,\Delta}(eq(\partial , y,
\cdot )),
\end{equation*}
defined on the Rota algebra $(\mathcal{F},+,*_{\Delta})$. Once represented
this equation on the lattice $\mathcal{L}$, it reduces to eq.~\eqref{eqnonlin}. Besides, the series expansion of the solution
$u_{n}$ truncates and converts into the finite sum  \eqref{part}.
\end{proof}

The same scheme of Remark~\ref{rem10} can be applied to generalize  Theorem~\ref{main2} to arbitrary objects of
$\mathcal{N}_{\{m; c_{0},\ldots , c_{N}\}}$.

\begin{remark}
\label{rem4}
A direct comparison between the main results of the linear and nonlinear
cases can be easily obtained. By putting $a_{N}(x)=1$, and
$a_{1}(x)=\ldots = a_{n-1}(x)=0$ into  eq.~\eqref{lincont2}, it converts
into  eq.~\eqref{ncont2} for $c_{2}(x)=\ldots =c_{N}(x)=0$. We observe that,
coherently,  eq.~\eqref{eqlin2} converts into  eq.~\eqref{eqnonlin}.
\end{remark}

\section{Integrable dynamics in the Fourier space}
\label{sec:6}

The integrable maps arising from the preceding construction induce an auxiliary
dynamics on the space of their Fourier coefficients, which is of independent
interest. We consider here the case of nonlinear equations with constant
coefficients, which provides a prototypical illustration of this alternative
construction.
\begin{proposition}
\label{prop2}
 Consider a dynamical system of the form
\begin{equation}
\frac{d^{m}}{dx^{m}}y= c_{N} y^{N}+c_{N-1} y^{N-1}+\ldots +c_{1} y+ c_{0},
\label{ncont3}
\end{equation}
where $c_{0},\ldots ,c_{N} \in \mathbb{R}$. Assume that
\begin{equation}
y=\sum _{k=0}^{\infty} \zeta _{k} x^{k}
\label{eq6.2}
\end{equation}
be a real solution of  \eqref{ncont3}, analytic at $x=0$. Then
$\zeta _{k}$ satisfies the associated equation
\begin{eqnarray}
\nonumber
\frac{(n+m)!}{n!}\zeta _{n+m}&=&c_{N}\sum _{
\overset{l_{1},\ldots ,l_{N-1}=0}{l_{1}+\ldots + l_{N-1}\leq n}}
\zeta _{l_{1}} \cdots \zeta _{l_{N-1}} \zeta _{n-l_{1}-\ldots -l_{N-1}}
+\ldots
\\
&+&c_{2} \sum _{l_{1}=0}^{n}\zeta _{l_{1}} \zeta _{n-l_{1}}+c_{1}
\zeta _{n} +c_{0}\delta_{n,0}.
\label{eq6.3}
\end{eqnarray}
\end{proposition}
\begin{proof}
As a consequence of the definition of basic sequence for
$\{p_{k}\}_{k\in \mathbb{N}}$, one can easily prove that, for each
$m\in \mathbb{N}\setminus \{0\}$
\begin{equation}
\Delta ^{m} u_{n}=\sum _{l=0}^{n}\frac{n!}{(n-l)!} \frac{(l+m)!}{l!}
\zeta _{l+m}.
\label{eq6.4}
\end{equation}
Also,
\begin{equation}
u_{n}^{*p}=\sum _{l_{1},\ldots ,l_{p}=0}^{\infty}\zeta _{l_{1}}
\ldots \zeta _{l_{p}}p_{(l_{1}+\ldots +l_{p})}(n)=\sum _{l=0}^{n}
\frac{n!}{(n-l)!}\cdot \sum _{
\overset{l_{1},\ldots ,l_{p-1}=0}{l_{1}+\ldots + l_{p-1}\leq l}}
\zeta _{l_{1}} \cdots \zeta _{l_{p-1}} \zeta _{l-l_{1}-\ldots -l_{p-1}}.
\label{eq6.5}
\end{equation}
Combining the previous expressions yields the stated result.
\end{proof}
Both the linear case and the case of dynamical systems with nonconstant
coefficients can be treated analogously. Their detailed analysis is left
to the reader.

\section{Rota algebras and  Picard-Vessiot theory}
\label{sec:7}

This section aims to extend certain fundamental results of Galois theory,
established for homogeneous linear differential equations with constant
coefficients, to their \textit{alter ego} on $\mathcal{L}$ as defined by
Theorem~\ref{main1} (See e.g. \cite{Franke}, \cite{Kolchin}, \cite{PS1997}, \cite{PS2003}, \cite{Serre}  for relevant
definitions and results).

Let $\{a_{0},\ldots ,a_{N-1}\}\subset \mathbb{R}$. Let
$\mathcal{B(\mathcal{F}, \mathcal{Q})}$ be the space of linear operators
acting on the Rota differential algebra
$(\mathcal{F},+, \cdot , *_{\mathcal{Q}})$ over $\mathbb{R}$. We introduce
the linear operator
$T[\mathcal{Q}]\in \mathcal{B(\mathcal{F}, \mathcal{Q})}$ defined by
$T[\mathcal{Q}]:= \mathcal{Q}^{N}+a_{N-1} \mathcal{Q}^{N-1}+\ldots +a_{1}
\mathcal{Q}+a_{0}$. We consider the linear differential equation
\begin{equation}
T[{\partial}](y):= y^{(N)}+a_{N-1} y^{(N-1)}+\ldots +a_{1}y'+a_{0} y =0,
\label{diff}
\end{equation}
and its categorically equivalent equation
\begin{equation}
T[{\Delta}](u):=\Delta ^{N}u+a_{N-1} \Delta ^{N-1}u+\ldots +a_{1}
\Delta u +a_{0} u=0.
\label{discr}
\end{equation}
One can show that the rings $C_{\partial}$ and $C_{\Delta}$ of constants for $\mathcal{R}(\mathcal{F}, \partial)$ and $\mathcal{R}(\mathcal{F}, \Delta)$  coincide.
 The following result holds.
\begin{lemma}
\label{fundamental}
The morphism
$\rho _{\partial ,\Delta}: (\mathcal{F}, \mathcal{\partial})
\longrightarrow (\mathcal{F}, \Delta )$ maps isomorphically a fundamental
system of solutions $\mathcal{S}$ of the linear differential equation \eqref{diff} into a fundamental system
$\rho _{\partial , \Delta}(\mathcal{S})$ of the linear  equation \eqref{discr}.
\end{lemma}
\begin{proof}
As a consequence of Theorem~\ref{main1}, power series solutions of eqs.~\eqref{diff}
and \eqref{discr} are in one-to-one correspondence. The morphism
$\rho _{\partial ,\Delta}$ maps basic sequences into basic sequences, so
it preserves linear independence of solutions and the dimension of the
associated vector space. Then
$\mathcal{S}'=\rho _{\partial , \Delta}(\mathcal{S})$ is a fundamental
set for eq.~\eqref{discr}.
\end{proof}

Our purpose will be to define an extension of a differential subring $\mathcal{F}'\subseteq \mathcal{F}$ for eq.~\eqref{discr} by using the categorical approach developed before. To this aim, we construct the \textit{universal solution algebra} $\mathcal{U}$ of eq.~\eqref{discr}.
Let
\begin{equation}
\Delta Y =AY,\qquad A\in M_N(\mathbb{R})
\label{matrixdiff}
\end{equation}
be its matrix form, where it is understood that
$Y_{i+1,j}=\Delta Y_{ij}$. We introduce an $N\times N$ matrix of indeterminates
$Y=(Y_{ij})$ and define
\begin{equation}
\mathcal{U}[\Delta ]:=\mathcal{F}'\left [Y_{ij},
\frac{1}{{\det}_* Y}\right ], \quad 1\leq i,j\leq N.
\label{eq7.4}
\end{equation}
Here we introduce the modified Wronskian
\begin{equation}
\label{wronsk}
{\det}_* Y:=
\left|\begin{matrix}
Y_{11} &Y_{12}&\ldots & Y_{1N}
\\
\Delta Y_{11}&\Delta Y_{12}&\ldots &\Delta Y_{1N}
\\
\Delta ^{2} Y_{11}&\ldots & & \Delta ^{2} Y_{1N}
\\
\vdots & & &\vdots
\\
\Delta ^{N-1}Y_{11}&\ldots & &\Delta ^{N-1} Y_{1N}
\end{matrix}\right|
\ .
\end{equation}

\begin{definition}
A Picard--Vessiot (PV) ring for the matrix equation \eqref{matrixdiff} over the Rota algebra $(\mathcal{F}',\Delta )$ is a simple differential ring $\mathcal{O}(+,\cdot,*_{\Delta})$ such that

i) There exists a fundamental matrix
$Z\in \text{GL}_{N}(\mathcal{O})$ for eq.~\eqref{matrixdiff}.

ii) The ring $\mathcal{O}$ is generated by $\mathcal{F}'$, the entries of
$Z$ and $\frac{1}{{\det}_* Z}$.
\end{definition}
To construct a PV-ring, we introduce the notion of \textit{Rota ideal}.
\begin{definition}
\label{defn11}
A \textit{Rota ideal} for the derivation $\Delta$ is an ideal invariant under  $\Delta$.
\end{definition}
The following result holds.
\begin{proposition}
\label{prop3}
Let $\mathcal{I}$ be a maximal Rota ideal of
$\mathcal{U}[\Delta ]$. Then the ring
$\mathcal{V}[\Delta ]:=\mathcal{U}[\Delta ]/\mathcal{I}$ is a
\textit{Picard-Vessiot ring} for the equation \eqref{matrixdiff}.
\end{proposition}
\begin{proof}
Because $\mathcal{I}$ is a maximal Rota ideal, the quotient ring $\mathcal{V}[\Delta]$ is a simple differential ring under the induced action of $\Delta$. Then, it suffices to observe that in the quotient ring the Wronskian \eqref{wronsk} is invertible.
\end{proof}
We shall denote by $\mathcal{U}[\partial ]$ the universal solution algebra of the matrix differential
equation $Y'=AY$ associated to eq.~\eqref{diff}, where
$Y_{i+1,j}= Y'_{i,j}$. We introduce the notion of a differential Galois
group for the previous equations.

\begin{definition}
\label{defn12}
Given the Rota algebra $(\mathcal{F}',\Delta )$, let
$\mathcal{V}[\Delta ]$ be a Picard-Vessiot ring over $\mathcal{F}'$. The
differential Galois group of $\mathcal{V}[\Delta ]$ over
$(\mathcal{F}',\Delta )$,
$\text{DGal}\left (\mathcal{V}/ (\mathcal{F}',\Delta )\right )$, is the group
of the differential $(\mathcal{F}',\Delta )$-automorphisms.
\end{definition}
Let $\mathcal{M}$ denote a maximal differential ideal of eq.~\eqref{diff}. The main result of this section is the following

\begin{theorem}
\label{thm4}
Let $\mathcal{W}[\partial ]:=\mathcal{U}[\partial ]/\mathcal{M}$ be a PV-ring over $(\mathcal{F}', \partial)$
for the linear homogeneous differential equation  \eqref{diff} and
$\mathcal{V}[\Delta ]=\mathcal{U}[\Delta ]/\mathcal{I}$ be the associated PV-ring over $(\mathcal{F}', \Delta)$ for the
equation  \eqref{discr}. Then there exists an isomorphism of groups $\Phi $ such that
\begin{equation}
\text{DGal}\left (\mathcal{W}[\partial]/ (\mathcal{F}', \partial )\right )
\stackrel{\Phi}{\longleftrightarrow} \text{DGal}\left (\mathcal{V}[
\Delta ]/ (\mathcal{F}', \Delta )\right ).
\label{eq7.7}
\end{equation}
\end{theorem}
\begin{proof}
According to Lemma~\ref{fundamental}, given a fundamental set
$\mathcal{S}$ of solutions of eq.~\eqref{diff},
$\mathcal{S}'=\rho _{\partial , \Delta}(\mathcal{S})$ is a fundamental
set of solutions of eq.~\eqref{discr}. This implies that the universal
solution algebras $\mathcal{U}[\partial ]$ and
$\mathcal{U}[\Delta ]$ are isomorphic. The morphism
$\lambda _{\partial , \Delta}$ maps a (maximal) differential ideal
$\mathcal{M}$ into a (maximal) Rota ideal $\mathcal{I}$, so that
the associated Picard-Vessiot rings are isomorphic. Then the differential
automorphisms of the Galois groups associated with eqs.~\eqref{diff} and \eqref{discr} are in one-to-one correspondence.
\end{proof}
\begin{remark}
\label{rem5}
An analogous result holds for the  equation
$T[\nabla ](u)=0$, where $\nabla =1-T^{-1}$. However, the previous construction
does not obviously extend to an arbitrary delta operator
$\mathcal{Q}$. In fact, the fundamental set of equation $T[\mathcal{Q}](u)=0$ generates a linear space whose dimension is not necessarily preserved.
\end{remark}

\section{Main theorems for the singular case}
\label{sec:8}

In this Section we shall prove the main theorems of our theory of Laurent
polynomial sequences.

In the following Theorems~\ref{main3} and \ref{main4}, we will define our
maps in the \textit{twisted Rota algebra}
$(\mathcal{H}, \mathcal{Q})$, where $\mathcal{H}$ is the space of Laurent
formal series, and $p_{k}(x)_{k\in \mathbb{Z}}$ are
\textit{Laurent basic sequences} for the delta operator $Q=\Delta $.

\subsection{Linear difference equations}
\label{sec8.1}

\begin{theorem}
\label{main3}
Consider an ordinary linear differential equation of order
$N\in \mathbb{N}\setminus \{0\}$ of the form
\begin{equation}
\label{eq:main3}
lin(\partial , y):=a_{N}(x) \frac{d^{N}}{dx^{N}}y + a_{N-1}(x)
\frac{d^{N-1}}{dx^{N-1}}y +\ldots +a_{1}(x) \frac{d}{dx} y+a_{0}(x) y +g(x)=0,
\end{equation}
where $a_{i}(x)$, $i=0,\ldots ,N$ and $g(x)$ are real functions admitting at
most polar singularities at $x=0$ of the form:
\begin{equation}
\label{eq:8.2}
a_{i}(x)=\sum _{j=-\ell _{i}}^{r_{i}}\alpha _{ij} x^{j}, \qquad g(x)=
\sum _{j=-\ell}^{r}\gamma _{j} x^{j} .
\end{equation}
Assume that eq.~\eqref{eq:main3} has a solution possessing a polar singularity
of order $s$ at the origin:
\begin{equation}
\label{eq:lpolsol}
y(x)=\sum _{k=1}^{s}\frac{\zeta _{-k}}{x^{k}}
\end{equation}
where
$s\geq \max \{\ell +r,\ell _{0}+r_{0},\ldots ,\ell _{N}+r_{N} \}$. Then,
the difference equation
\begin{equation}
\label{eq:8.4}
\sum _{m=0}^{N} \sum _{j=-\ell _{m}}^{r_{m}} \alpha _{mj}
\frac{(n-1)!}{(n-j-1)!}\Delta ^{m}u_{n-j}+\sum _{j=-\ell}^{r}\gamma _{j}
p_{j}(n)=0,
\end{equation}
lying in the positive real line, has a solution of the form
\begin{equation}
\label{eq:8.5}
u_{n}=\sum _{k=1}^{s}\frac{(n-1)!}{(n+k-1)!}\zeta _{-k}.
\end{equation}
\end{theorem}
\begin{proof}
Our aim is to construct the equation categorically equivalent to eq.~\eqref{eq:main3}
in $\mathcal{K}_{\{ a_{0},\ldots , a_{N},g\}}$, for the class of functions \eqref{eq:lpolsol}. Consequently, we represent the product
$x^{r}\frac{\mathrm{d}^{s}}{\mathrm{d}x^{s}}y(x)$ on the set
$\mathcal{L}$ by the action of the morphism
$\rho _{\partial ,\Delta}: (\mathcal{H},+,\cdot )\rightarrow (
\mathcal{H},+,*_{\Delta})$. This morphism determines the correspondences
detailed below.

Let us first consider the case of the higher derivatives
\begin{equation}
y^{(m)}=\sum _{k=-\infty}^{\infty}\frac{k!}{(k-m)!}\zeta _{k}x^{k-m}
\longrightarrow \sum _{k=-\infty}^{\infty}\frac{k!}{(k-m)!}\zeta _{k} p_{k-m}(n),
\label{eq8.6}
\end{equation}
\begin{equation}
\Delta ^{m}u_{n}=\sum _{k=-\infty}^{\infty} \zeta _{k} \Delta ^{m}p_{k
}(n)=\sum _{k=-\infty}^{\infty} \frac{k!}{(k-m)!}\zeta _{k} p_{k-m}(n),
\label{eq8.7}
\end{equation}
and then,
\begin{equation}
y^{(m)}\longrightarrow \Delta ^{m} u_{n}= \sum _{i=0}^{m}(-1)^{m-i}
\binom{m}{i}u_{n+i}.
\label{eq8.8}
\end{equation}
These results are a natural consequence of the properties of the delta
operators. By analogy, we get
\begin{eqnarray*}
x^{r} y^{(m)} &\longrightarrow & p_{r}(n)*\sum _{k=-\infty}^{0}
\frac{k!}{(k-m)!}\zeta _{k} p_{k-m}(n)=\sum _{k=-\infty}^{0}
\frac{k!}{(k-m)!} \zeta _{k} p_{k-m+r}(n)
\\
&=&\sum _{k=-\infty}^{-1} \frac{k!}{(k-m)!}\zeta _{k}
\frac{(n-1)!(n-r-k+m-1)!}{(n-r-1)!(n-k+m-r-1)!}p_{k-m}(n-r),
\\
&=&\frac{(n-1)!}{(n-r-1)!}\sum _{k=-\infty}^{-1} \frac{k!}{(k-m)!}
\zeta _{k} p_{k-m}(n-r)
\\
&=&\frac{(n-1)!}{(n-r)!}\bigg((n-r)\sum _{k=-\infty}^{-1}
\frac{k!}{(k-m)!} \zeta _{k} p_{k-m}(n-r)\bigg)
\\
&=&\frac{(n-1)!}{(n-r-1)!}\Delta ^{m}u_{n-r}.
\end{eqnarray*}
Therefore, taking into account relations  \eqref{eq:8.2}, we obtain
\begin{equation}
a_{m}(x)y^{(m)}(x) \to \sum _{j=-\ell _{m}}^{r_{m}} \alpha _{mj}
\frac{(n-1)!}{(n-j-1)!}\Delta ^{m}u_{n-j}.
\label{eq8.9}
\end{equation}
By combining the previous formulas, we deduce the form of eq.~\eqref{eq:8.4}.

To prove that the series \eqref{eq:lpolsol} is solution of eq.~\eqref{eq:8.4}, observe that the morphism
$\lambda _{\partial ,\Delta}$ provides the correspondence
\begin{equation}
\sum _{k} \zeta _{-k} x^{-k} \longrightarrow \sum _{k} \zeta _{-k}
\frac{1}{p_{k}^{-}(n)}.
\label{eq8.10}
\end{equation}

The action of the functor $F$ carries any solution $y$ of eq.~\eqref{eq:main3},
defined on the algebra $(\mathcal{H},+,\cdot )$, into a solution
$u_{n}$ of the corresponding equation
\begin{equation*}
lin(\Delta ,u,*_{\Delta})=\lambda _{\partial ,\Delta}(lin(\partial , y,
\cdot )),
\end{equation*}
defined on the extended Rota algebra $(\mathcal{H},+,*_{\Delta})$. Once
represented this equation on the lattice $\mathcal{L}$, any finite Laurent
series solution  \eqref{eq:lpolsol} of  eq.~\eqref{eq:main3} is carried into
a solution of the form \eqref{eq:8.5} of the equation
\begin{equation*}
lin(\mathcal{Q},u,*_{\mathcal{Q}})=\lambda _{\partial ,\Delta}(lin(
\partial ,y,\cdot )).
\end{equation*}
This equation, according to the previous construction, coincides with eq.~\eqref{eq:8.4}.
\end{proof}

\subsection{Nonlinear difference equations}
\label{sec8.2}

We shall prove a theorem concerning the discretization of nonlinear equations
of the form  \eqref{nonlincont} which preserves solutions admitting polar
singularities.
\begin{theorem}
\label{main4}
Consider an ordinary differential equation of order
$m\in \mathbb{N}\setminus \{0\}$ of the form
\begin{equation}
\label{eq:main4}
eq(\partial , y):=\frac{d^{m}}{dx^{m}}y-\sum _{r=0}^{N} c_{r}(x) y^{r}=0,
\end{equation}
where $N, m\in \mathbb{N}\setminus \{0\}$, $c_{r}(x)$,
$r=0,\ldots ,N$ are real functions admitting at most polar singularities at
$x=0$ of the form:
\begin{equation}
c_{r}(x)=\sum _{j=-\ell _{r}}^{\rho _{r}}\beta _{rj} x^{j}.
\label{eq8.12}
\end{equation}
Assume that eq.~\eqref{eq:main4} has a solution possessing a polar singularity
of order $s$ at the origin:
\begin{equation}
\label{eq:npolsol}
y(x)= \sum _{k=1}^{s}\frac{\zeta _{-k}}{x^{k}},
\end{equation}
where
$s\geq \max \{\ell _{0}+\rho _{0},\ldots ,\ell _{N}+\rho _{N} \}$. Then,
the difference equation
\begin{equation}
\label{eq:8.14}
\sum _{i=0}^{m}(-1)^{m-i}\binom{m}{i}u_{n+i}=\sum _{r=0}^{N} \sum _{j=-
\ell _{r}}^{\rho _{r}} \sum _{k_{1},\cdots ,k_{r}=-\infty}^{-1}
\beta _{rj} \, \zeta _{k_{1}}\cdots \zeta _{k_{r}}p_{K_r+j}(n)
\end{equation}
with $K_{r}=\sum _{i=1}^{r} k_{i}$, defined over a uniform lattice in the positive real line, has a
solution of the form
\begin{equation}
\label{eq:8.15}
u_{n}=\sum _{k=1}^{s}\frac{(n-1)!}{(n+k-1)!}\zeta _{-k}.
\end{equation}
\end{theorem}
\begin{proof}
\label{eq:8.23}
The discretization runs as follows. The function $y^{r}(x)$, is written
as a Laurent expansion centred at the origin. Let us write formally
\begin{equation}
y^{r}\to \left (\sum _{k_{1}=-\infty}^{-1}\zeta _{k_{1}} p_{k_{1}}(n)
\right )*\cdots * \left (\sum _{k_{r}=-\infty}^{-1}\zeta _{k_{r}} p_{k_{r}}(n)
\right )= \sum _{k_{1},\cdots ,k_{r}=-\infty}^{-1}\zeta _{k_{1}}
\cdots \zeta _{k_{r}}p_{K_{r}}(n),
\label{eq8.16}
\end{equation}
where $K_{r}=\sum _{i=1}^{r} k_{i}$, and $\zeta _{k_{i}}=0$ for
$k_{i}=-s-1,\ldots , -\infty $. Thus,
\begin{eqnarray}
c_{r}(x) y^{r} &\to & \sum _{j=-\ell _{r}}^{\rho _{r}}\beta _{rj} p_{j}(n)*
\sum _{k_{1},\cdots ,k_{r}=-\infty}^{-1}\zeta _{k_{1}}\cdots \zeta _{k_{r}}p_{K_{r}}(n)
=
\label{eq8.17}
\\
\nonumber
&=&\sum _{j=-\ell _{r}}^{\rho _{r}} \sum _{k_{1},\cdots ,k_{r}=-
\infty}^{-1} \beta _{rj} \, \zeta _{k_{1}}\cdots \zeta _{k_{r}}p_{K_{r}+j}(n).
\end{eqnarray}
Then, the discretized function, $u_{n}$, is written as
\begin{equation}
u_{n}=\sum _{k=1}^{s}\frac{\zeta _{-k}}{p^{-}_{k}(n)}.
\label{eq8.18}
\end{equation}
The derivatives $y^{(m)}$, $m=1,\ldots , r$, are discretized according
to the operator $\Delta $:
\begin{equation}
\Delta ^{m} u_{n}=\sum _{i=0}^{m}(-1)^{i}\binom{m}{i}u_{n+i}.
\label{eq8.19}
\end{equation}
The constants $\zeta _{-k}$ satisfy a linear algebraic system:
\begin{equation}
\label{eq:syst}
\begin{cases}
u_{n+1} = \sum _{k=1}^{s}\frac{\zeta _{-k}}{p^{-}_{k}(n+1)}= \sum _{k=1}^{s}
A_{n+1,k}\,\zeta _{-k},
\\
\hspace{3mm}
\vdots
\\
u_{n+s} = \sum _{k=1}^{s}\frac{\zeta _{-k}}{p^{-}_{k}(n+s)}= \sum _{k=1}^{s}
A_{n+s,k}\,\zeta _{-k},
\end{cases}
\end{equation}
where $A_{j k}=\frac{(j-1)!}{(j+k-1)!}$, $k=1\ldots ,s$,
$j=n+1\ldots ,n+s$. Unlike in the previous theorems, where the corresponding
coefficient matrix acquires a triangular form, in the present case the
solution of system  \eqref{eq:syst} cannot be easily written in a closed
form in terms of the variables $u_{n+i}$. We obtain
\begin{equation}
\zeta _{-k}= \sum _{\lambda =1}^{s} \big(A^{-1}\big)_{k,\lambda}u_{n+
\lambda}.
\label{eq8.21}
\end{equation}
Thus, the discrete equation reads:
\begin{eqnarray}\nn
\sum _{i=0}^{m}(-1)^{i}\binom{m}{i}u_{n+i} = \sum _{r=0}^{N} \sum _{j=-
\ell _{r}}^{\rho _{r}} \sum _{k_{1},\cdots ,k_{r}=-\infty}^{-1} \sum _{
\lambda _{1},\ldots ,\lambda _{r}=1}^{s} \beta _{rj} \big(A^{-1}\big)_{k_{1},
\lambda _{1}}\cdots \big(A^{-1}\big)_{k_{r},\lambda _{r}} \cdot
\label{eq8.22}
\\
\cdot \,\, u_{n+\lambda _{1}}\cdots u_{n+\lambda _{r}}p_{K_{r}+j}(n),
\end{eqnarray}
where
\begin{equation}
p_{K_r+j}(n)=
\begin{cases}
\dfrac{n!}{(n-K_r-j)!}, & K_r+j\geq 0,
\\[10pt]
\dfrac{(n-1)!}{(n-K_r-j-1)!}, & K_r+j <0.
\end{cases}
\label{eq8.23}
\end{equation}

Let us prove that the series  \eqref{eq:8.15} is solution of eq.~\eqref{eq:8.14}. To this aim, we observe that the morphism
$\rho _{\partial ,\Delta}$ carries on
\begin{equation}
\sum _{k=1}^{s} \zeta _{-k} x^{-k} \longrightarrow \sum _{k=1}^{s}
\zeta _{-k} \frac{1}{p_{k}^{-}(n)}.
\label{eq8.24}
\end{equation}
Thus, the action of the functor $F$ carries any solution $y$ of eq.~\eqref{eq:main3},
defined on the algebra $(\mathcal{H},+,\cdot )$, into a solution
$u_{n}$ of the corresponding equation
\begin{equation*}
eq(\Delta ,u,*_{\Delta})=\nu _{\partial ,\Delta}(eq(\partial , y,
\cdot )),
\end{equation*}
defined on the extended Rota algebra $(\mathcal{H},+,*_{\Delta})$. According
to the previous analysis, we represent this equation on the lattice
$\mathcal{L}$. Then, any finite Laurent series solution  \eqref{eq:npolsol} of  eq.~\eqref{eq:main4} is carried into a solution of
the form  \eqref{eq:8.15} of the equation
\begin{equation*}
eq(\mathcal{Q},u,*_{\mathcal{Q}})=\nu _{\partial ,\Delta}(eq(
\partial ,y,\cdot )).
\end{equation*}
This equation coincides with eq.~\eqref{eq:8.14}.
\end{proof}

\begin{remark}
\label{rem6}
We wish to emphasize that, in concrete examples, the applicability of the
categorical discretization method proposed in Theorems~\ref{main3} and \ref{main4} can be considerably extended. While generalizing the hypotheses
of our theorems may not be straightforward, specific cases can nevertheless
be analyzed and solved rigorously.
\end{remark}

\section{Nonlocal Lie symmetries and integrable maps: some conjectures}
\label{sec9}

The categorical discretization introduced in Theorems~\ref{main1},
\ref{main2}, \ref{main3} and \ref{main4} establishes a precise structural
correspondence between the space of smooth solutions of the original differential
equations and the space of solutions of the associated difference equations.
This result ensures that the passage from the continuous to the discrete
setting, carried out within this categorical framework, does not merely
approximate the original system but preserves its essential solution structure
in a rigorous and exact sense.

Building on this observation, it is a natural problem to ascertain to what degree such a correspondence, established at the level of solutions,
extends to the symmetry properties of the systems under consideration.
Specifically, we hypothesize the existence of a meaningful relationship
between the Lie symmetry groups and their corresponding Lie algebras admitted
by ordinary differential equations and those admitted by their discrete counterparts.
For a comprehensive account of the modern theory of Lie symmetries for
continuous systems, see e.g. \cite{olver}. If confirmed, this would suggest
that the discretization not only preserves solutions but also retains,
in an appropriate sense, the geometric and algebraic structures encoded
by symmetries, which play a central role in the qualitative analysis of
differential equations.

We first conjecture the existence of symmetry transformations.
\begin{conjecture}
\label{conj1}
The integrable maps  \eqref{eqlin} and  \eqref{eq:8.4} (resp.~\eqref{eqnonlin}
and  \eqref{eq:8.14} for the singular case) admit a Lie group
$\mathcal{G}$ of nonlocal diffeomorphisms that leave the map invariant
and transform solutions into solutions.
\end{conjecture}
This conjecture relies on the idea that, for the ODEs considered, the determining
linear equations generating the Lie symmetries in the continuous case would
be preserved by our discretization scheme, as well as their solutions.
The nonlocal character of the diffeomorphisms generating these symmetries
is suggested by the intrinsic nonlocality of our $*$-product.

However, in general, objects categorically equivalent do not possess isomorphic
fundamental solution sets, with the exception of linear equations with
constant coefficients, as discussed in Section~\ref{sec:6}. Consequently, the full
Lie algebras generated by the symmetries postulated in the previous conjecture,
are, in general, not isomorphic. The subsequent conjecture establishes
a weaker form of correspondence between these algebras.

\begin{conjecture}
\label{conj2}
The Lie algebra of the generators of the Lie group $\mathcal{G}$ of the
nonlocal symmetry diffeomorphisms associated with the maps  \eqref{eqlin}  and  \eqref{eq:8.4} (resp.~\eqref{eqnonlin} and  \eqref{eq:8.14}) contains a subalgebra which is isomorphic to the Lie algebra
of the classical Lie point symmetries of the continuous dynamical system  \eqref{lincont} (resp.~\eqref{nonlincont}).
\end{conjecture}

\section{New integrable maps from ODEs with regular solutions}
\label{sec:9}

As an illustration of the usefulness of the categorical approach developed
in this work, we propose here some examples of integrable maps constructed
according to the main  Theorems~\ref{main1} and \ref{main2}, concerning
differential equations admitting regular solutions. These maps are associated
with ODEs that are particularly relevant in the applications.

\subsection{A discrete damped harmonic oscillator}
\label{sec10.1}

The damped harmonic oscillator
\begin{equation}
\label{damp}
y''(x)+2 q \omega y'(x)+\omega ^{2} y(x)=0
\end{equation}
in the case $q< 1$ admits the general solution
\begin{equation}
y_{\mathrm{gen}}(x)=A e^{-q \omega x}\sin \left ( \sqrt{1-q^{2}} \,
\omega x+ \phi \right ) \ .
\label{eq10.2}
\end{equation}
The \textit{discrete damped oscillator} reads
\begin{equation}
u_{n+2} +2(q \omega -1) u_{n+1} +(\omega ^{2}-2q\omega +1)u_{n} = 0
\ ,
\label{eq10.3}
\end{equation}
with the general solution
\begin{equation}
u_{n}=\sum _{k=0}^{n}g(k)\frac{n!}{(n-k)!},
\qquad
g(k):=\frac{1}{k!}\left [\frac{d}{dx^{k}} y_{\mathrm{gen}}\right ]_{x=0}
\ .
\label{eq10.4}
\end{equation}
We shall consider as lattice sites the points $nh$,
$h\in \mathbb{R}$. In this mesh, the operator $\Delta _{h}$ is defined
as $\Delta _{h} u_{n}=\frac{u_{n+1}-u_{n}}{h}$, with basic polynomials
\begin{equation}
\label{eq:10.5}
p_{k}(x)=\prod _{j=0}^{k-1}(x-jh),
\end{equation}
possessing zeros at the sites of the lattice. Thus:
\begin{equation}
p_{k}(nh)=
\begin{cases}
0 & \mathrm{if}\; n<k,
\\
\dfrac{n!h^{k}}{(n-k)!} & \mathrm{if}\; n \geq k
\label{pol2}
\end{cases}
\end{equation}
and the inverse interpolating transform reads
\begin{equation}
\label{eq:6.29}
u_{n}= \sum _{k=0}^{n} \frac{n!h^{k}}{(n-k)!} \zeta _{k},\quad \zeta _{k}=
\frac{1}{h^{k}}\sum _{j=0}^{k} \frac{(-1)^{k-j}}{j!(k-j)!}u_{j}.
\end{equation}
Therefore, the discrete damped oscillator equation is:
\begin{equation}
u_{n+2} +2(h q \omega -1) u_{n+1} +(h^{2}\omega ^{2}-2 h q\omega +1)u_{n}
= 0.
\label{eq10.8}
\end{equation}
Its continuum limit (when $h\to 0$, $n\to \infty $ and $nh$ remains bounded)
is the equation  \eqref{damp}, as can be easily checked.

\subsection{A difference equation for the Gaussian function}
\label{sec10.2}

 We discretize the differential equation satisfied by the Gaussian function:
\begin{equation}
\label{eq:9.1}
y'(x)=-x y(x).
\end{equation}
The solution (up to a factor) has the following Taylor expansion:
\begin{equation}
\label{eq:gauss}
y(x)=e^{-\frac{x^{2}}{2}}=\sum _{k=0}^{\infty}
\frac{(-1)^{k}}{2^{k}k!}x^{2k}.
\end{equation}

According to  Theorem~\ref{main1}, in a uniform lattice of step size
$h$ the associated integrable map reads
\begin{equation}
\label{eq:diff1}
u_{n+1}-u_{n}+h^{2} nu_{n-1}=0.
\end{equation}
It admits the solution ($u_{0}=1$):
\begin{equation}
\label{eq:9.3}
u_{n}=\sum _{k=0}^{\infty}\frac{n!\,h^{k}\,\zeta _{k}}{(n-k)!}=\sum _{k=0}^{[n/2]}
\frac{(-1)^{k} n!\, h^{2 k}}{2^{k} k!\,(n-2 k)!},
\end{equation}
where the coefficients $\zeta _{k}$ are given by the Taylor expansion in  \eqref{eq:gauss}. Here $[n/2]$ denotes the integer part of $n/2$.

\begin{remark}
\label{rem7}
The purpose of this simple analysis is to illustrate, in a transparent
way and using a basic example, the non-standard form of the difference
equations obtained through our categorical discretization. In fact, a classical
discretization of eq.~\eqref{eq:9.1} would lead to the difference equation
\begin{equation}
\label{eq:diff2}
\hat{u}_{n+1}-\hat{u}_{n}+h^{2}\,n\,\hat{u}_{n}=0.
\end{equation}
However, it is straightforward to verify that this equation does not admit
a solution of the form  \eqref{eq:9.3}, which arises directly from the Gaussian
function. The equation  \eqref{eq:diff2} admits an exact solution which
can be easily obtained ($\hat{u}_{0}=1$):
\begin{equation}
\hat{u}_{n}=\prod _{k=0}^{n-1}(1-k\,h^{2}),\quad n\ge 1
\label{eq10.14}
\end{equation}

In both cases (categorical and standard discretizations), the resulting
difference equations  \eqref{eq:diff1} and  \eqref{eq:diff2} tend to the
differential equation  \eqref{eq:gauss} in a straightforward way (we divide
by $h$ and take the continuum limit, $h\to 0$, $n\to \infty $, keeping
$n h$ constant). It can also be verified that the discretized solutions
provide a very good approximation of the exponential function.
\end{remark}

\subsection{Hypergeometric equation}
\label{sec10.3}

The real hypergeometric differential equation
\begin{equation}
x(1-x)\frac{d^{2} y}{dx^{2}}+\left [c-(a+b+1)x\right ]\frac{dy}{dx}-ab
y=0,
\label{hde}
\end{equation}
where $a,b,c\in \mathbb{R}$, possesses three regular singular points
$\{0,1, \infty \}$ \cite{AS}. By way of example, we shall restrict to the
singularity at $x=0$. Around this point,  eq.~\eqref{hde} admits two algebraically
independent solutions,
\begin{equation}
s_{1}:={}_{2} F_{1}(a,b;c;x),
\label{eq10.16}
\end{equation}
and a second solution $s_{2}$ whose explicit form depends on the particular
values of the parameters.  Theorem~\ref{main1} provides the following discrete
version of  eq.~\eqref{hde}
\begin{eqnarray}
&&(n+c) u_{n+1}-(n^{2}+(a+b+2)n+ab+c)u_{n}
\label{eq10.17}
\\
&&\qquad +n( a+b+2n )u_{n-1} -n(n-1) u_{n-2}
\nonumber
=0.
\nonumber
\label{hyperd}
\end{eqnarray}
If we introduce the finite Gauss sum
\begin{equation}
G(a,b;c;n):=\sum _{k=0}^{n}\frac{(a)_{k} (b)_{k}}{(c)_{k}}
\frac{n!}{(n-k)!},
\label{eq10.18}
\end{equation}
the solution of  eq.~\eqref{hyperd} corresponding to $s_{1}$ is provided
by
\begin{equation}
u^{(1)}_{n}:=G(a,b;c;n).
\label{eq10.19}
\end{equation}

A numerical analysis of this solution shows that this solution is in excellent
agreement with the continuous one.

Theorem~\ref{main1} allows us to introduce integrable maps which discretize
the differential equations characterizing families of orthogonal polynomials.
Here we shall focus on the nontrivial case of the equation defining the
classical family of Jacobi.

\subsection{A discrete Jacobi equation}
\label{sec10.4}

Applying Theorem~\ref{main1}, the equation
\begin{equation}
(1-x^{2})y''(x)+\left [\beta -\alpha -(\alpha +\beta +2)x\right ]y'(x)+m(m+
\alpha +\beta +1)y(x)=0.
\end{equation}
can be discretized in the form
\begin{eqnarray}
\nn
u_{n+2} - (\alpha -\beta +2)u_{n+1} +\left (m^{2}-n^{2}+(m-n) (
\alpha +\beta +1) + \alpha -\beta +1\right )u_{n}
\label{eq10.21}\quad
\\
+ n (\alpha +\beta +2 n)u_{n-1} -n(n-1) u_{n-2}=0.
\end{eqnarray}

This equation admits as a solution the family of polynomials
$\{P_{m}^{(\alpha ,\beta )}(n)\}_{m\in \mathbb{N}}$, where
\begin{eqnarray*}
P_{m}^{(\alpha ,\beta )}(n)&:=&
\frac{\Gamma (\alpha +m+1)}{m!\Gamma (\alpha +\beta +m+1)}
\\
&\times & \sum _{k=0}^{m}\binom{m}{k}
\frac{\Gamma (\alpha +\beta +m+k+1)}{\Gamma (\alpha +k+1)}
\frac{n! }{2^{k}}\sum _{i=0}^{k} \binom{k}{i}
\frac{(-1)^{k-i} }{(n-i)!}.
\nonumber
\end{eqnarray*}
Virtually all classes of orthogonal polynomials, including nonclassical
and Sobolev-type ones, can be associated with integrable maps in an analogous
way. The analysis of other similar cases is left to the reader.

\subsection{A new discrete Painlev\'e I equation}
\label{sec10.5}

Let us discuss the interesting case of the discretization of the classical
Painlev\'e I equation. Let us start with the differential equation
\begin{equation}
\frac{d^{2} y}{dy^{2}} - c_{2}(x) y^{2}- c_{1}(x) y - c_{0}(x) =0.
\label{eq10.22}
\end{equation}
According to Theorem~\ref{main2}, for $m=2$, and $N=2$, the associated
difference equation is:
\begin{eqnarray}
&&u_{n+2}-2u_{n+1}+u_{n}
\nonumber
\\
&& - \sum _{s,j_{1},j_{2}=0}^{n}\beta _{2s}
\frac{(-1)^{n-s -j_{1}-j_{2} } n!}{ (n-s-j_{1}-j_{2})! }
\frac{u_{j_{1}}u_{j_{2}}}{j_{1}!j_{2}!}
\nonumber
\\
&& - \sum _{s,j_{1}=0}^{n}\beta _{1s} \frac{ n!}{ (n-s-j_{1})! }
\frac{u_{j_{1}}}{j_{1}!} - \sum _{s=0}^{n}\beta _{0s }
\frac{n!}{(n-s )!} =0.
\label{eq10.23}
\end{eqnarray}

In particular, for the Painlev\'e I equation
\begin{equation}
\label{eq:PainI}
y''=6y^{2}+x,
\end{equation}
we have
\begin{equation}
c_{0}(x)=x,\quad c_{1}(x)=0,\quad c_{2}(x)=6.
\label{eq10.25}
\end{equation}
As is well-known, the Painlev\'e I equation admits an analytic solution
around $x=0$, with a local expansion given by
\begin{equation}
y(x)= \zeta _{0} + \zeta _{1} x+3 \zeta _{0}^{2} x^{2}+ \Big(2 \zeta _{0}
\zeta _{1}+\frac{1}{6}\Big) x^{3}+ \Big(\frac{1}{2} \zeta _{1} ^{2}+ 3
\zeta _{0}^{3}\Big) x^{4} + \cdots
\label{eq10.26}
\end{equation}
Here $\zeta _{0}=y(0)$ and $\zeta _{1} = y'(0)$.

Consequently, by applying  Theorem~\ref{main1}, we can introduce the novel
\textit{discrete Painlev\'e I equation}
\begin{equation}
\label{eq:57}
u_{n+2}-2u_{n+1}+u_{n} - 6\sum _{j_{1},j_{2}=0}^{n}
\frac{(-1)^{n -j_{1}-j_{2} } n!}{ (n-j_{1}-j_{2})! }
\frac{u_{j_{1}}u_{j_{2}}}{j_{1}!j_{2}!} - n=0.
\end{equation}
To check the continuum limit of  eq.~\eqref{eq:57}, we introduce in the
semi-axis $x\ge 0$ a uniform mesh $\mathcal{L}_{h}$ of points
\begin{equation}
x_{n}=nh,\quad n\in \mathbb{N},
\label{eq10.28}
\end{equation}
where $h>0$. The continuum limit is defined by $h\to 0$,
$n\to \infty $ and $nh$ bounded.

By using formulas  \eqref{eq:10.5}--\eqref{eq:6.29}, we rewrite eq.
\eqref{eq:57} in terms of the mesh amplitude $h$. We get
\begin{equation}
\label{eq:58}
\frac{1}{h^{2}}\left (u_{n+2}-2u_{n+1}+u_{n}\right ) - 6\sum _{j_{1},j_{2}=0}^{n}
\frac{(-1)^{n -j_{1}-j_{2} } n!}{ (n-j_{1}-j_{2})! }
\frac{u_{j_{1}}u_{j_{2}}}{j_{1}!j_{2}!} - nh=0.
\end{equation}
A straightforward analysis shows that in the double limit $h\to 0$,
$n\to \infty $ ($nh$ bounded),  eq.~\eqref{eq:58} reduces to the Painlev\'e
I equation  \eqref{eq:PainI}.

The first equations $n=0,1,2,3$ are explicitly given by:
\begin{eqnarray}
&&-6 h^{2} u_{0}^{2}+u_{0}-2 u_{1}+u_{2}=0,
\label{sys}
\\
\nonumber
&&-h^{3}-6 h^{2} u_{0} (2 u_{1}-u_{0})+u_{1}-2 u_{2}+u_{3}=0,
\\
\nonumber
&&-2 h^{3}-6 h^{2} \left (u_{0}^{2}+(2 u_{2}-4 u_{1}) u_{0}+2 u_{1}^{2}
\right )+u_{2}-2 u_{3}+u_{4}=0,
\\
\nonumber
&&-3 h^{3}+6 h^{2} \left (u_{0}^{2}-2 (3 u_{1}-3 u_{2}+u_{3}) u_{0}+6 u_{1}
(u_{1}-u_{2})\right )+u_{3}-2 u_{4}+u_{5}.
\end{eqnarray}
An exact solution of this equation is:
\begin{equation}
u_{n}=\sum _{k=0}^{n}\zeta _{k}p_{k}(n)=n!\sum _{k=0}^{n}
\frac{h^{k}\zeta _{k}}{(n-k)!}.
\label{eq10.31}
\end{equation}
Since $\zeta _{0}=u_{0}$, $\zeta _{1}=(u_{1}-u_{0})/h$, we have
\begin{eqnarray*}
u_{2}&=&6 h^{2} u_{0}^{2}-u_{0}+2 u_{1},
\\
u_{3}&=&h^{3}+6 h^{2} u_{0} (u_{0}+2 u_{1})-2 u_{0}+3 u_{1},
\\
u_{4}&=&72 h^{4} u_{0}^{3}+4 h^{3}+12 h^{2} u_{1} (2 u_{0}+u_{1})-3 u_{0}+4
u_{1}.
\end{eqnarray*}
One can easily check that these functions satisfy the system \eqref{sys}.

\section{New integrable maps admitting singular solutions}
\label{sec:10}

We shall construct several new maps obtained as an application of the results
proposed in Theorems~\ref{main3} and \ref{main4}.

\subsection{A linear model with two polar singularities}
\label{seis}

The differential equation
\begin{equation}
y'' +\frac{3}{x} y'+\frac{1}{x} y=\frac{1}{x^{2}}
\label{eq11.1}
\end{equation}
admits a general solution which can be written in terms of Bessel functions,
together with a particular solution of the inhomogeneous equation, which
may be chosen as:
\begin{equation}
\label{sol3}
y(x)=\frac{1}{x}+\frac{1}{x^{2}}.
\end{equation}
The corresponding difference equation can be easily obtained from  Theorem~\ref{main3}. We have
\begin{equation}
u_{n+2}-2u_{n+1}+u_{n}+\frac{3}{n}(u_{n+2}-u_{n+1})+\frac{1}{n }u_{n+1}
-\frac{1}{n(n+1)}=0.
\label{eq11.3}
\end{equation}
It admits the solution
\begin{equation}
u_{n}=\frac{n+2}{n(n+1)},
\label{eq11.4}
\end{equation}
constructed directly from solution \eqref{sol3}:
\begin{equation}
\frac{1}{x}+\frac{1}{x^{2}}\to p_{-1}(n)+p_{-2}(n)=\frac{1}{n}+
\frac{1}{n(n+1)}.
\label{eq11.5}
\end{equation}
If we introduce a step $h$ for the lattice, the equation is:
\begin{equation}
(n+1) \left [(n+3 )u_{n+2}-(2n+3-h) u_{n+1}+ n u_{n}\right ] =1,
\label{eq11.6}
\end{equation}
with solution
\begin{equation}
u_{n}=\frac{(n+1)h+1}{n (n+1) h^{2}}.
\label{eq11.7}
\end{equation}

\subsection{A nonlinear dynamical system admitting a triple-pole solution}
\label{sec11.2}

Consider the nonlinear equation:
\begin{equation}
\label{cont}
(ax+b)^{2}y''-6x(ax+2b)\, y^{2} =0, \qquad a,b\in \mathbb{R}
\end{equation}
which admits the solution
\begin{equation}
\label{eq:10.16}
y(x)= \frac{a}{x^{2}}+ \frac{b}{x^{3}}.
\end{equation}
For sake of clarity, we shall derive its discrete version in a thorough
way. This solution has the discrete counterpart given by
\begin{equation}
y(x)\to u_{n}=\sum _{k=-3}^{-2}\zeta _{k}p_{k}(n).
\label{eq11.10}
\end{equation}
Thus, we deduce the following relations.
\begin{eqnarray*}
xy^{2} &\to & p_{1}(n)*\left (\sum _{k=-3}^{-2}\zeta _{k}p_{k}(n)
\right )*\left (\sum _{k=-3}^{-2}\zeta _{k}p_{k}(n)\right )
\\
&=&\zeta _{-2}^{2}p_{-3}(n)+2\zeta _{-2}\zeta _{-3}p_{-4}(n)+\zeta _{-3}^{2}p_{-5}(n),
\\
x^{2}y^{2} &\to & p_{2}(n)*\left (\sum _{k=-3}^{-2}\zeta _{k}p_{k}(n)
\right )*\left (\sum _{k=-3}^{-2}\zeta _{k}p_{k}(n)\right )
\\
&=&\zeta _{-2}^{2}p_{-2}(n)+2\zeta _{-2}\zeta _{-3}p_{-3}(n)+\zeta _{-3}^{2}p_{-4}(n),
\\
xy'' &\to & p_{1}(n)*\sum _{k=-3}^{-2}k(k-1)\zeta _{k}p_{k-2}(n) = (n-1)(u_{n+1}-2u_{n
}+u_{n-1}),
\\
x^{2}y'' &\to & p_{2}(n)*\sum _{k=-3}^{-2}k(k-1)\zeta _{k}p_{k-2}(n)= (n-1)(n-2)(u_{n}-2u_{n-1}+u_{n-2}).
\end{eqnarray*}

The coefficients $\zeta _{-2}$ and $\zeta _{-3}$ can be written in terms
of $u_{n}$, $u_{n+1}$ solving the linear system:
\begin{equation}
\begin{cases}
u_{n}=\zeta _{-3} p_{-3}(n)+\zeta _{-2} p_{-2}(n)
\\
u_{n+1}=\zeta _{-3} p_{-3}(n+1)+\zeta _{-2} p_{-2}(n+1)
\end{cases}
,\quad p_{k}(n)=\frac{(n-1)!}{(n-k-1)!},\quad k< 0
\label{eq11.11}
\end{equation}
whose solution reads
\begin{equation}
\begin{cases}
\zeta _{-2}= - n(n+1) (n+2) u_{n}+(n+1)(n+2)(n+3) u_{n+1}
\\
\zeta _{-3}= n(n+1) (n+2) (n+3) u_{n}-(n+1) (n+2)^{2} (n+3) u_{n+1}.
\end{cases}
\label{eq11.12}
\end{equation}
Then, the equation \eqref{cont} admits the following categorically equivalent
discretization:
\begin{eqnarray}
\nonumber
&&b^{2} u_{n+2} +\frac{6(n+1)(n+2)(n+3)(a (n+2) (n+4)+2 b n)}{n(n+4)}u_{n+1}^{2}
+2b(a(n-1)-b)u_{n+1}
\\
\nonumber
&-&\frac{12(n+1)(n+2)(a(n+3)(n+4)+2 b(n+2))}{n+4}u_{n} u_{n+1} \\
\nonumber
&+&
\frac{6n(n+1)(n+2)(a (n+4)+2 b)}{n+4}u_{n}^{2}
\\
\nonumber
&+&(a(n-1)(a(n-2)-4b)+b^{2})u_{n}-2a(n-1)(a(n-2)-b)u_{n-1} \\
\nonumber
&+&a^{2}(n-2)(n-1)u_{n-2}
=0.
\end{eqnarray}

It is straightforward to prove that it admits the categorical counterpart
of solution  \eqref{eq:10.16}:
\begin{equation}
u_{n}=\frac{a}{n(n+1)}+\frac{b}{n(n+1)(n+2)}.
\label{eq11.13}
\end{equation}

\section{Future perspectives}
\label{sec12}

The categorical discretization approach developed in this article is primarily
intended to provide an efficient method for constructing integrable maps
from integrable continuous systems. However, it should be emphasized that,
although the general framework introduced here finds its natural origin
in the theory of integrable systems, it is not restricted to this specific
field.

In fact, it can, in principle, be regarded as a general discretization
approach that remains applicable even when exact solutions of the underlying
ODEs are not explicitly known. In such cases, the resulting difference
equations may be interpreted as discrete approximations of the continuous
models, and can thus be used to construct approximate solutions.

In this article, we have focused on the case of scalar ODEs as a first
essential step toward further generalizations. At present, we do not foresee
any technical obstacles to extending the framework to vector-valued ODEs,
and we intend to pursue this direction in future work. Extending the categorical
approach to partial differential equations also constitutes one of the
main objectives of our ongoing research. This fundamental problem could
in principle be addressed using the same ideas,
\textit{mutatis mutandis}.

The relation between the present approach and the long-standing problem
of the symmetry-preserving discretization of differential equations constitutes
an intriguing and worthwhile line of research.

A thorough investigation of this relationship will be carried out in future
research, encompassing both ordinary differential equations (ODEs) and
partial differential equations (PDEs) (for instance, by extending some of the results of \cite{Tempesta2} to the categorical framework) in order to assess its theoretical
implications as well as its practical applicability.

An interesting direction for future research is the extension of the present
theory to the case of $q$-difference equations. A coherent Galois approach
to this class of equations has been developed in \cite{diV},
\cite{Sauloy}. We hypothesize that a functorial correspondence can also
be envisaged for this case, allowing a unified treatment of the Galois
theory for differential equations and their categorically equivalent
$q$-difference equations.

As a broader objective of our work, we conjecture that the study of quantum
field theories on a lattice - one of the most significant developments
in high-energy physics in recent decades - could be further enhanced through
a categorical approach. Work along these research directions is currently
underway.

\appendix
\section{Some combinatorial relations}

\label{AppendixA}

The following identities are used in the proofs of the main theorems:
\begin{equation}
\sum _{k=0}^{n}(-1)^{k}\binom{n}{k} =\delta _{n0},\quad \sum _{k=0}^{r}
\binom{s}{k}\binom{m}{r-k}=\binom{m+s}{r},
\label{eqA.1}
\end{equation}
\begin{equation}
\label{delt}
\sum _{k=0}^{s}(-1)^{s-k}\binom{s}{k}\binom{n+k}{m+s}=\binom{n}{m},
\end{equation}
\begin{equation}
\sum _{k=j}^{n+s-r} (-1)^{k-j} \binom{n-r}{k-s}\binom{k}{j}=\sum _{i=0}^{s}(-1)^{s-i}
\binom{s}{i}\delta _{n-r+i,j}.
\label{eqA.3}
\end{equation}

\section{The inverse transform}
\label{AppendixB}

We revise here the construction of the inverse transform. Consider the
$\Delta $ operator and its sequence of basic polynomials computed in the
lattice $\mathcal{L}$. If
\begin{equation}
u_{n}=\sum _{k=0}^{n} \zeta _{k}p^{+}_{k}(n)=\sum _{k=0}^{n}
\frac{n!\zeta _{k}}{(n-k)!},
\label{eqB.1}
\end{equation}
then the Fourier coefficients $\zeta _{k}$ are given by
\cite{TempestaJDE}:
\begin{equation}
\label{eq:B2}
\zeta _{k}=\sum _{j=0}^{k} \frac{(-1)^{k-j}u_{j}}{j!(k-j)!}.
\end{equation}

A crucial point in the proof of this property is that the expression of
$u_{n}$ in terms of the sequence of basic polynomials has a finite number
of addends in the sum, since these polynomials have their zeros at the
lattice points. However, when we consider a series of Laurent polynomials,
the expansions have generically an infinite number of terms (in the Laurent
part):
\begin{equation}
u_{n}=\sum _{k=1}^{\infty} \frac{\zeta _{-k}}{p^{-}_{k}(n)}+\sum _{k=0}^{n}
\zeta _{k} p^{+}_{k}(n).
\label{eqB.3}
\end{equation}
This formula cannot be inverted as in the case of  eq.~\eqref{eq:B2} and,
consequently, the coefficients $\zeta _{-k}$ do not have a general expression
as functions of $u_{n}$.

\section{Categorical correspondences}
\label{AppendixC}

Here we show explicitly some technical details, relevant in the proofs
of the main theorems, of the way the categorical correspondence acts on
several algebraic and differential structures.

\bigskip
\noindent
\textit{Differentials and products}
\begin{eqnarray}
x^{r}\frac{\mathrm{d}^{s}}{\mathrm{d}x^{s}}y(x)&=&\sum _{k=0}^{\infty}
\frac{k!}{(k-s)!} \zeta _{k} x^{k+r-s} \longrightarrow \sum _{k=0}^{n-r+s}
\frac{k!}{(k-s)!}\zeta _{k} p_{k+r-s}(nh)
\nonumber
\\
&=&\sum _{k=0}^{n-r+s}
\frac{n!k!h^{k+r-s} \zeta _{k} }{(k-s)!(n-k-r+s)!}.
\label{eqC.1}
\end{eqnarray}
Using the identity  \eqref{delt} we get
\begin{eqnarray}
&&x^{r}\frac{\mathrm{d}^{s}}{\mathrm{d}x^{s}}y(x) \longrightarrow
\frac{ n!h^{r-s}}{(n-r)!}\sum _{k=0}^{n-r+s}\sum _{i=0}^{s}(-1)^{s-i}k!
\binom{s}{i}\binom{n-r+i}{k} \zeta _{k}
\label{eqC.2}
\\
&=&
\nonumber
\frac{ n!h^{r-s}}{(n-r)!}\sum _{i=0}^{s}(-1)^{s-i}\binom{s}{i}\sum _{k=0}^{n-r+i}
\frac{(n-r+i)!}{ (n-r+i-k)!} \zeta _{k}
\\
&=&
\nonumber
\frac{ n!h^{r }}{(n-r)!}\frac{1}{h^{s}}\sum _{i=0}^{s}(-1)^{s-i}
\binom{s}{i}u_{n-r+i}=\frac{ n!h^{r}}{(n-r)!}\Delta ^{s}_{h}u_{n-r}.
\end{eqnarray}
\textit{Proof of equations  \eqref{eq:3.24} and  \eqref{eq:3.25} in Theorem~\ref{main2}}.

Let us write the term $y(x)^{r}$ as a multiple series. We get
\begin{equation}
\label{eq:B3}
y(x)^{r}=\sum _{k_{1}=0}^{\infty}x^{k_{1}}\sum _{k_{2}=0}^{k_{1}}
\zeta _{k_{1}-k_{2}} \sum _{k_{3}=0}^{k_{2}}\zeta _{k_{2}-k_{3}}
\cdots \sum _{k_{r-1}=0}^{k_{r-2}}\zeta _{k_{r-2}-k_{r-1}}\sum _{k_{r}=0}^{k_{r-1}}
\zeta _{k_{r-1}-k_{r}} \zeta _{k_{r}}.
\end{equation}
We will express the coefficients $\zeta _{k}$ in terms of $u_{i}$ by means
of the inverse transforms of eqs. \eqref{eq:6.29}, suitably iterated. First,
we observe that there is no dependence from the mesh spacing $h$, since
the terms involving $h$ satisfy the identity
\begin{equation}
h^{k_{1}}\frac{1}{h^{k_{1}-k_{2}}}\frac{1}{h^{k_{2}-k_{3}}}\cdots
\frac{1}{h^{k_{r-2}-k_{r-1}}}\frac{1}{h^{k_{r-1}-k_{r}}}
\frac{1}{h^{k_{r}}}=1.
\label{eqC.4}
\end{equation}

Starting from the last sum in  \eqref{eq:B3}, we get
\begin{eqnarray}
&&\sum _{k_{r}=0}^{k_{r-1}}\zeta _{k_{r-1}-k_{r}}\zeta _{k_{r}}
\label{eqC.5}
\\
&\to & \sum _{k_{r}=0}^{k_{r-1}}\sum _{j_{r-1}=0}^{k_{r-1}-k_{r}}
\sum _{j_{r}=0}^{k_{r}}
\frac{(-1)^{k_{r-1}-j_{r-1}-k_{r}}}{j_{r-1}!(k_{r-1}-j_{r-1}-k_{r})!}
\frac{(-1)^{k_{r}-j_{r }}}{j_{r}!(k_{r}-j_{r})!}u_{j_{r-1}}u_{j_{r}}
\nonumber
\\
&=& \sum _{j_{r-1}=0}^{k_{r-1}}\sum _{k_{r}=0}^{k_{r-1}-j_{r-1}}\sum _{j_{r}=0}^{k_{r}}
\frac{(-1)^{k_{r-1}-j_{r-1}-j_{r} }}{j_{r-1}!j_{r}!(k_{r-1}-j_{r-1}-k_{r})!(k_{r}-j_{r})!}
u_{j_{r-1}}u_{j_{r}}
\nonumber
\\
&=& \sum _{j_{r-1}=0}^{k_{r-1}}\sum _{j_{r}=0}^{k_{r-1}-j_{r-1}}
\frac{(-1)^{k_{r-1}-j_{r-1}-j_{r} }u_{j_{r-1}}u_{j_{r}}}{j_{r-1}!j_{r}!}
\nonumber
\\
&\times &\sum _{k_{r}=j_{r}}^{k_{r-1}-j_{r-1}}
\frac{1}{(k_{r-1}-j_{r-1}-k_{r})!(k_{r}-j_{r})!}.
\nonumber
\end{eqnarray}
Since
\begin{equation}
\sum _{k_{r}=j_{r}}^{k_{r-1}-j_{r-1}}
\frac{1}{(k_{r-1}-j_{r-1}-k_{r})!(k_{r}-j_{r})!}=
\frac{2^{k_{r-1}- j_{r-1}-j_{r}}}{( k_{r-1}-j_{r-1}-j_{r})!},
\label{eqC.6}
\end{equation}
then we have
\begin{equation}
\label{term1}
\sum _{k_{r}=0}^{k_{r-1}}\zeta _{k_{r-1}-k_{r}}\zeta _{k_{r}}\to
\sum _{j_{r-1}=0}^{k_{r-1}}\sum _{j_{r}=0}^{k_{r-1}-j_{r-1}}
\frac{(-1)^{k_{r-1}-j_{r-1}-j_{r} }2^{k_{r-1}- j_{r-1}-j_{r}}u_{j_{r-1}}u_{j_{r}}}{j_{r-1}!j_{r}!( k_{r-1}-j_{r-1}-j_{r})!}.
\end{equation}
The sums in the next term
\begin{eqnarray}
\label{eq:B8}
&&\sum _{k_{r-1}=0}^{k_{r-2}} \sum _{j_{r-2}=0}^{k_{r-2}-k_{r-1}}
\frac{(-1)^{k_{r-2}-j_{r-2}-k_{r-1}}}{j_{r-2}!(k_{r-2}-j_{r-2}-k_{r-1})!}u_{j_{r-2}}
\\
&\times & \sum _{j_{r-1}=0}^{k_{r-1}} \sum _{j_{r}=0}^{k_{r-1}-j_{r-1}}
\frac{(-1)^{k_{r-1}-j_{r-1}-j_{r} }2^{k_{r-1}- j_{r-1}-j_{r}}u_{j_{r-1}}u_{j_{r}}}{j_{r-1}!j_{r}!( k_{r-1}-j_{r-1}-j_{r})!}
\nonumber
\end{eqnarray}
can be appropriately ordered in the following way;
\begin{eqnarray}
&&\sum _{k_{r-1}=0}^{k_{r-2}} \sum _{j_{r-2}=0}^{k_{r-2}-k_{r-1}}
\sum _{j_{r-1}=0}^{k_{r-1}} \sum _{j_{r}=0}^{k_{r-1}-j_{r-1}} = \sum _{j_{r-2}=0}^{k_{r-2}}
\sum _{j_{r-1}=0}^{k_{r-2}-j_{r-2}} \sum _{j_{r}=0}^{k_{r-2}-j_{r-2}-j_{r-1}}
\sum _{k_{r-1}=j_{r}+j_{r-1}}^{k_{r-2}-j_{r-2}}.
\nonumber
\end{eqnarray}
Thus, we can rewrite the expression  \eqref{eq:B8} as
\begin{eqnarray}
\label{eq:B9}
&&\sum _{j_{r-2}=0}^{k_{r-2}} \sum _{j_{r-1}=0}^{k_{r-2}-j_{r-2}}
\sum _{j_{r}=0}^{k_{r-2}-j_{r-2}-j_{r-1}}
\frac{(-1)^{k_{r-2}-j_{r-2}-j_{r-1}-j_{r}}u_{j_{r-2}}u_{j_{r-1}}u_{j_{r}}}{j_{r-2}!j_{r-1}!j_{r}!}
\\
&\times & \sum _{k_{r-1}=j_{r}+j_{r-1}}^{k_{r-2}-j_{r-2}}
\frac{2^{k_{r-1}-j_{r-1
}-j_{r}}}{( k_{r-1}-j_{r-1}-j_{r})!(k_{r-2}-j_{r-2}-k_{r-1})!}.
\nonumber
\end{eqnarray}
The last sum can be easily computed:
\begin{equation}
\sum _{k_{r-1}=j_{r}+j_{r-1}}^{k_{r-2}-j_{r-2}}
\frac{2^{k_{r-1}-j_{r-1
}-j_{r}}}{( k_{r-1}-j_{r-1}-j_{r})!(k_{r-2}-j_{r-2}-k_{r-1})!}=
\frac{3^{k_{r-2}-j_{r-2}-j_{r-1
}-j_{r}}}{(k_{r-2}-j_{r-2}-j_{r-1}-j_{r})!}
\label{eqC.10}
\end{equation}
and substituting it into  eq.~\eqref{eq:B9}, we obtain the term:
\begin{equation}
\label{term2}
\sum _{j_{r-2}=0}^{k_{r-2}} \sum _{j_{r-1}=0}^{k_{r-2}-j_{r-2}} \sum _{j_{r}=0}^{k_{r-2}-j_{r-2}-j_{r-1}}
\frac{(-1)^{k_{r-2}-j_{r-2}-j_{r-1}-j_{r}}3^{k_{r-2}-j_{r-2}-j_{r-1
}-j_{r}}u_{j_{r-2}}u_{j_{r-1}}u_{j_{r}}}{j_{r-2}!j_{r-1}!j_{r}!(k_{r-2}-j_{r-2}-j_{r-1}-j_{r})!}.
\end{equation}

The comparison between the relations  \eqref{term1} and  \eqref{term2} allows
us to deduce, by analogy, the general relation
\begin{eqnarray}
\\
\nn
&&\sum _{k_{l}=0}^{k_{l-1}}\zeta _{k_{l-1}-k_{l}}\cdots \sum _{k_{r-1}=0}^{k_{r-2}}
\zeta _{k_{r-2}-k_{r-1}}\sum _{k_{r}=0}^{k_{r-1}}\zeta _{k_{r-1}-k_{r}}
\zeta _{k_{r}} =
\label{eqC.12}
\\
\nonumber
&=&\sum _{j_{l-1}=0}^{k_{l-1}} \sum _{j_{l}=0}^{k_{l-1}-j_{l-1}}
\cdots \sum _{j_{r-1}=0}^{k_{l-1}-\sum_{i=l-1}^{r-2}j_i} \;\;\sum _{j_{r}=0}^{k_{l-1}-\sum_{i=l-1}^{r-1}j_i}
\frac{(l-r-2)^{k_{l-1}-\sum _{i=l-1}^{r}j_{i}}}{( k_{l-1}-\sum _{i=l-1}^{r}j_{i})!}
\prod _{i=l-1}^{r} \frac{u_{j_{i}}}{j_{i}!}.
\end{eqnarray}
In particular, for $l=2$, we obtain
\begin{equation}
\sum _{j_{1}=0}^{k_{1}} \sum _{j_{2}=0}^{k_{1}-j_{1}} \cdots \sum _{j_{r-1}=0}^{k_{1}-\sum _{i=1}^{r-2}j_{i}}
\sum _{j_{r}=0}^{k_{1}-\sum _{i=1}^{r-1}j_{i}}
\frac{(-r)^{k_{1}-\sum _{i=1}^{r}j_{i}}}{( k_{1}-\sum _{i=1}^{r}j_{i})!}
\prod _{i=1}^{r} \frac{u_{j_{i}}}{j_{i}!}.
\label{eqC.13}
\end{equation}
By moving the first sum to the end, we arrive at the correspondence
\begin{equation}
y(x)^{r}\to \sum _{k_{1}=0}^{n}\frac{n!}{(n-k_{1})!} \sum _{j_{1}=0}^{k_{1}}
\sum _{j_{2}=0}^{k_{1}-j_{1}} \cdots \sum _{j_{r-1}=0}^{k_{1}-\sum _{i=1}^{r-2}j_{i}}\;\;
\sum _{j_{r}=0}^{k_{1}-\sum _{i=1}^{r-1}j_{i}}
\frac{(-r)^{k_{1}-\sum _{i=1}^{r}j_{i}}}{( k_{1}-\sum _{i=1}^{r}j_{i})!}
\prod _{i=1}^{r} \frac{u_{j_{i}}}{j_{i}!}.
\label{eqC.14}
\end{equation}
Finally, since
\begin{eqnarray}
&&\sum _{k_{1}=0}^{n} \sum _{j_{1}=0}^{k_{1}} \sum _{j_{2}=0}^{k_{1}-j_{1}}\;\;\sum _{j_{3}=0}^{k_{1}-(j_{1}+j_2)}
\cdots \sum _{j_{r-1}=0}^{k_{1}-\sum _{i=1}^{r-2}j_{i}}\;\; \sum _{j_{r}=0}^{k_{1}-\sum _{i=1}^{r-1}j_{i}} =
\label{eqC.15}
\\
&=& \sum _{j_{1}=0}^{n} \sum _{j_{2}=0}^{n-j_1}\sum _{j_{3}=0}^{n-(j_1+j_2)} 
\cdots \sum _{j_{r-1}=0}^{n-\sum _{i=1}^{r-2}j_{i}} \;\;\sum _{j_{r}=0}^{n-\sum _{i=1}^{r-1}j_{i}} \sum _{k_{1}=\sum _{i=1}^{r}{j_{i}}}^{n}
\nonumber
\end{eqnarray}
we are led to the sum
\begin{equation}
\sum _{k_{1}=\sum _{i=1}^{r}{j_{i}}}^{n}
\frac{(-r)^{k_{1}-\sum _{i=1}^{r}j_{i}}}{(n-k_{1})!( k_{1}-\sum _{i=1}^{r}j_{i})!}=
\frac{(1-r)^{n-\sum _{i=1}^{r}j_{i}}}{(n-\sum _{i=1}^{r}j_{i})!}.
\label{eqC.16}
\end{equation}
Consequently, collecting together all the intermediate steps, our final
result is:
\begin{equation}
y(x)^{r}\to \sum _{j_{1}=0}^{n}\;\; \sum _{j_{2}=0}^{n-j_1}\sum _{j_{3}=0}^{n-(j_1+j_2)} \cdots \sum _{j_{r-1}=0}^{n-\sum _{i=1}^{r-2}j_{i}}\;\;\sum _{j_{r}=0}^{n-\sum _{i=1}^{r-1}j_{i}}
\frac{n!(1-r)^{n-\sum _{i=1}^{r}j_{i}}}{(n-\sum _{i=1}^{r}j_{i})!}
\prod _{i=1}^{r} \frac{u_{j_{i}}}{j_{i}!}.
\label{eqC.17}
\end{equation}
By following the same technique, the series expansion of the function
$x^{m}y(x)^{r}$ can be written as:
\begin{equation}
x^{m}y(x)^{r}=\sum _{k_{1}=0}^{\infty}x^{k_{1}+m}\sum _{k_{2}=0}^{k_{1}}
\zeta _{k_{1}-k_{2}} \sum _{k_{3}=0}^{k_{2}}\zeta _{k_{2}-k_{3}}
\cdots \sum _{k_{r-1}=0}^{k_{r-2}}\zeta _{k_{r-2}-k_{r-1}}\sum _{k_{r}=0}^{k_{r-1}}
\zeta _{k_{r-1}-k_{r}}\zeta _{k_{r}}.
\label{eqC.18}
\end{equation}
Its categorically equivalent is:
\begin{eqnarray}
&&x^{m} y(x)^{r}\to \sum _{k_{1}=0}^{n-m}
\frac{n!h^{k_{1}+m}}{(n-m-k_{1})!}\sum _{k_{2}=0}^{k_{1}}\zeta _{k_{1}-k_{2}}
\sum _{k_{3}=0}^{k_{2}}\zeta _{k_{2}-k_{3}}\cdots
\label{eqC.19}
\\
&&\sum _{k_{r-1}=0}^{k_{r-2}}\zeta _{k_{r-2}-k_{r-1}}\sum _{k_{r}=0}^{k_{r-1}}
\zeta _{k_{r-1}-k_{r}}\zeta _{k_{r}}.
\nonumber
\end{eqnarray}
By using the previous results, we get:
\begin{eqnarray}
&&x^{m} y(x)^{r}\to
\label{eqC.20}
\\
&\to & h^{m} \sum _{k_{1}=0}^{n-m}\frac{n!}{(n-m-k_{1})!} \sum _{j_{1}=0}^{k_{1}}
\sum _{j_{2}=0}^{k_{1}-j_{1}} \cdots \sum _{j_{r-1}=0}^{k_{1}-j_{1}}
\sum _{j_{r}=0}^{k_{1}-j_{1}}
\frac{(-r)^{k_{1}-\sum _{i=1}^{r}j_{i}}}{( k_{1}-\sum _{i=1}^{r}j_{i})!}
\prod _{i=1}^{r} \frac{u_{j_{i}}}{j_{i}!}.
\nonumber
\end{eqnarray}
If we set $n-m$ as the $n$ in the previous section, as a final result we
obtain
\begin{equation}
x^{m}y(x)^{r}\to
\begin{cases}
0, & n< m
\\[8pt]
\displaystyle{ h^{m} \sum _{j_{1}=0}^{n-m} \sum _{j_{2}=0}^{n-m}
\cdots \sum _{j_{r}=0}^{n-m}
\frac{n!(1-r)^{n-m-\sum _{i=1}^{r}j_{i}}}{(n-m-\sum _{i=1}^{r}j_{i})!}
\prod _{i=1}^{r} \frac{u_{j_{i}}}{j_{i}!}},& n\ge m .
\end{cases}
\label{eqC.21}
\end{equation}

\section*{Acknowledgment}
The research of M.A.R. and P.T. has been supported by the Project PID2024-156610NB-I00
of Ministerio de Ciencia, Innovaci\'on y Universidades. P.T. has also
been supported by the Severo Ochoa Programme for Centres of Excellence in R\&D   (CEX-2023-001347-S), Ministerio de Ciencia, Innovaci\'{o}n y Universidades y Agencia Estatal de Investigaci\'on, Spain. P.T. is a member of the Gruppo
Nazionale di Fisica Matematica (GNFM) of the Istituto Nazionale di Alta
Matematica (INdAM).

\end{document}